\documentclass[twocolumn,pra,notitlepage,superscriptaddress,tightenlines,aps,10pt]{revtex4-1}
\usepackage{amsmath}
\usepackage{amssymb}
\usepackage{bbold}
\usepackage{amsthm}
\usepackage{graphicx}
\usepackage{cleveref}
\usepackage{dsfont}
\usepackage{xcolor}
\usepackage{tikz}
\usepackage{algpseudocode}
\usepackage{framed}
\usepackage{todonotes}
\usepackage{comment}
\usepackage{titlesec}
\DeclareMathOperator{\tr}{Tr}
\newcommand{\norm}[1]{\left\lVert#1\right\rVert}

\newcommand{\id}{\mathbb{1}}
\newcommand{\inp}[2]{\langle#1,#2\rangle}
\newcommand{\dens}[1]{|#1\rangle\!\langle#1|}
\newcommand{\ddens}[1]{|#1\rangle\!\rangle\!\langle\!\langle#1|}

\newcommand{\ket}[1]{|#1\rangle\!\rangle}
\newcommand{\bra}[1]{\langle\!\langle#1|}

\newcommand{\mc}[1]{\mathcal{#1}}

\newcommand{\md}[1]{\mathds{#1}}

\newcommand{\ct}{^\dagger}
\newcommand{\tn}[1]{^{\otimes #1}}

\newcommand{\gr}[1]{\ensuremath{{\bf \mathsf{#1}}}}
\DeclareMathOperator*{\avg}{ \raisebox{-3pt}{\text{\LARGE $\md{E}$}}}

\newtheorem{theorem}{Theorem}
\newtheorem{definition}{Definition}
\newtheorem{lemma}{Lemma}

\newcounter{notecounter}

\usepackage{natbib}
\usepackage{enumitem}
\usepackage{bm}
\usepackage{todonotes}

\newcommand{\M}{\mathrm{M}_{2^q}}
\newcommand{\CNOT}{\ensuremath{\mathrm{CNOT}}}

\newcommand{\bsq}{\boldsymbol{ \sigma }_q}
\date{\today}
\begin{document}
\title{A new class of efficient randomized benchmarking protocols}
\author{Jonas Helsen}
\email{j.helsen@tudelft.nl (corresponding author)}
\affiliation{QuTech, Delft University of Technology, Lorentzweg 1, 2628 CJ Delft, The Netherlands}
\author{Xiao Xue}
\email{x.xue@tudelft.nl}
\affiliation{QuTech, Delft University of Technology, Lorentzweg 1, 2628 CJ Delft, The Netherlands}
\affiliation{Kavli Institute of Nanoscience, Delft University of Technology, 2600 GA Delft, The Netherlands}
\author{Lieven M.K. Vandersypen}
\email{l.m.k.vandersypen@tudelft.nl}
\affiliation{QuTech, Delft University of Technology, Lorentzweg 1, 2628 CJ Delft, The Netherlands}
\affiliation{Kavli Institute of Nanoscience, Delft University of Technology, 2600 GA Delft, The Netherlands}
\author{Stephanie Wehner}
\email{s.d.c.wehner@tudelft.nl}
\affiliation{QuTech, Delft University of Technology, Lorentzweg 1, 2628 CJ Delft, The Netherlands}

\begin{abstract}
% Randomized benchmarking is an essential and widely used technique for reliably estimating average fidelities and other parameters of sets of quantum operations. Standard randomized benchmarking using the Clifford group yields data that can be fitted to a single exponential decay, allowing the average fidelity of the gateset to be extracted reliably.  However, protocols that go beyond standard randomized benchmarking, such as benchmarking gatesets with $T$-gates, simultaneous benchmarking or standard randomized benchmarking in the presence of leakage errors yield data that must be fitted to a linear combination of multiple exponential decays making it experimentally much harder to extract the average fidelity. In this letter we propose a new version of the randomized benchmarking protocol based on the theory of characters of representations that solves this problem with virtually no experimental overhead, making randomized benchmarking over more general gatesets experimentally viable. We illustrate the effectiveness of this new protocol by designing an interleaved benchmarking experiment that extracts the average fidelity of a $2$-qubit Clifford gate using only single-qubit Clifford gates as reference gates.

\noindent Randomized benchmarking is a technique for estimating the average fidelity of a set of quantum gates. However, if this gateset is not the multi-qubit Clifford group, robustly extracting the average fidelity is difficult. Here we propose a new method based on representation theory that has little experimental overhead and robustly extracts the average fidelity for a broad class of gatesets. We apply our method to a multi-qubit gateset that includes the $T$-gate, and propose a new interleaved benchmarking protocol that extracts the average fidelity of a two-qubit Clifford gate using only single-qubit Clifford gates as reference.

%96 words
\end{abstract}
 \maketitle
 \section*{Introduction}
Randomized benchmarking~\cite{dankert2006c,Magesan_2012,Emerson2005,Chow2009,Gaebler2012,Granade2014,Epstein2014} is arguably the most prominent experimental technique for assessing the quality of quantum operations in experimental quantum computing devices~\cite{Knill2008,Asaad2016,Chow2009,Barends2014,DiCarlo2009,o2015qubit,sheldon2016characterizing}. Key to the wide adoption of randomized benchmarking are its scalability with respect to the number of qubits and its insensitivity to errors in state preparation and measurement. It has also recently been shown to be insensitive to variations in the error associated to different implemented gates~\cite{wallman2018randomized,proctor2017randomized,merkel2018randomized}.

The randomized benchmarking protocol is defined with respect to a gateset $\gr{G}$, a discrete collection of quantum gates. Usually this gateset is a group, such as the Clifford group~\cite{Magesan_2012}. The goal of randomized benchmarking is to estimate the average fidelity~\cite{Nielsen2011} of this gateset.

Randomized benchmarking is performed by randomly sampling a sequence of gates of a fixed length $m$ from the gateset $\gr{G}$. This sequence is applied to an initial state $\rho$, followed by a global inversion gate such that in the absence of noise the system is returned to the starting state. Then the overlap between the output state and the initial state is estimated by measuring a two-component POVM $\{Q,\id-Q\}$. This is repeated for many sequences of the same length $m$ and the outputs are averaged, yielding a single average survival probability $p_m$. Repeating this procedure for various sequence lengths $m$ yields a list of probabilities $\{p_m\}_m$.

Usually $\gr{G}$ is chosen to be the Clifford group. It can then be shown (under the assumption of gate-independent CPTP noise)~\cite{Magesan2012a} that the data $\{p_m\}_m$ can be fitted to a single exponential decay of the form 
\begin{equation}\label{eq:rand_bench_fit}
p_m \approx_{\mathrm{fit}} A+ Bf^m
\end{equation}
where $A,B$ depend on state preparation and measurement, and the quality parameter $f$ only depends on how well the gates in the gateset $\gr{G}$ are implemented. This parameter $f$ can then be straightforwardly related to the average fidelity $F_{\mathrm{avg}}$~\cite{Magesan_2012}. The fitting relation \cref{eq:rand_bench_fit} holds intuitively because averaging over all elements 
of the Clifford group effectively depolarizes the noise affecting the input state $\rho$. This effective depolarizing noise then accretes exponentially with sequence length $m$.

However it is possible, and desirable, to perform randomized benchmarking on gatesets that are not the Clifford group,
and a wide array of proposals for randomized benchmarking using non-Clifford gatesets appear in the literature~\cite{Cross_2016,brown2018randomized,hashagen2018real,francca2018approximate,Dugas2015,Gambetta_2012_sim,harper2017estimating}. The most prominent use case is benchmarking a gateset $\gr{G}$ that includes the vital $T$-gate~\cite{Cross_2016,brown2018randomized,Dugas2015} which, together with the Clifford group, forms a universal set of gates for quantum computing~\cite{Nielsen2011}. Another use case is simultaneous randomized benchmarking~\cite{Gambetta_2012_sim}, which extracts information about crosstalk and unwanted coupling between neighboring qubits by performing randomized benchmarking on the gateset consisting of single qubit Clifford gates on all qubits. In these cases, and in other examples of randomized benchmarking with non-Clifford gatesets~\cite{hashagen2018real,Dugas2015,Gambetta_2012_sim}, the fitting relation~\cref{eq:rand_bench_fit} does not hold and must instead be generalized to 
\begin{equation}\label{eq:rand_bench_fit_gen}
p_m \approx_{\mathrm{fit}} \sum_{\lambda\in R_\gr{G}} A_\lambda f_\lambda^m,
\end{equation}
where $R_\gr{G}$ is an index set that only depends on the chosen gateset, the $f_{\lambda}$ are general `quality parameters' that only depend on the gates being implemented and the $A_\lambda$ prefactors depend only on SPAM (when the noise affecting the gates is trace preserving there will be a $\lambda\in R_{\gr{G}}$ -corresponding to the trivial subrepresentation- such that $f_{\lambda}=1$, yielding the constant offset seen in \cref{eq:rand_bench_fit}).
The above holds because averaging over sequences of elements of these non-Clifford groups averaging does not fully depolarize the noise. Rather the system state space will split into several ‘sectors’ labeled by $\lambda$, with a different depolarization rate, set by $f_\lambda$, affecting each sector.
The interpretation of the parameters $f_\lambda$ varies depending on the gateset $\gr{G}$. In the case of simultaneous randomized benchmarking~\cite{Gambetta_2012_sim} they can be interpreted as a measure of crosstalk and unwanted coupling between neighboring qubits. For other gatesets an interpretation is not always available. However, as was pointed out for specific gatesets in~\cite{Dugas2015,hashagen2018real,Cross_2016,brown2018randomized} and for general finite groups in~\cite{francca2018approximate}, the parameters $f_\lambda$ can always be jointly related (see \cref{eq:fid_gateset}) to the average fidelity $F_{\mathrm{avg}}$ of the gateset $\gr{G}$. This means that in theory randomized benchmarking can extract the average fidelity of a gateset even when it is not the Clifford group.

However in practice the multi-parameter fitting problem given by \cref{eq:rand_bench_fit_gen} is difficult to perform, with poor confidence intervals around the parameters $f_\lambda$ unless impractically large amounts of data are gathered. More fundamentally it is, even in the limit of infinite data, impossible to associate the estimates from the fitting procedure to the correct decay channel in \cref{eq:rand_bench_fit_gen} and thus to the correct $f_\lambda$, making it impossible to reliably reconstruct the average fidelity of the gateset.\\

In the current literature on non-Clifford randomized benchmarking, with the notable exception of~\cite{Dugas2015}, this issue is sidestepped by performing randomized benchmarking several times using different input states $\rho_\lambda$ that are carefully tuned to maximize one of the prefactors $A_\lambda$ while minimizing the others. This is unsatisfactory for several reasons: (1) the accuracy of the fit now depends on the preparation of $\rho_\lambda$, undoing one of the main advantages of randomized benchmarking over other methods such as direct fidelity estimation~\cite{Flammia2011}, and (2) it is, for more general gatesets, not always clear how to find such a maximizing state $\rho_\lambda$. These problems aren't necessarily prohibitive for small numbers of qubits and/or exponential decays (see for instance \cite{harper2019fault}) but they do limit the practical applicability of current non-Clifford randomized benchmarking protocols on many qubits and more generally restrict which groups can practically be benchmarked.\\

Here we propose an adaptation of the randomized benchmarking procedure, which we call character randomized benchmarking, which solves the above problems and allows reliable and efficient extraction of average fidelities for gatesets that are not the Clifford group. We begin by discussing the general method, before applying it to specific examples. Finally we discuss using character randomized benchmarking in practice and argue the new method does not impose significant experimental overhead. Previous adaptations of randomized benchmarking, as discussed in \cite{Knill2008,Muhonen2015,helsen2017multi} and in particular \cite{Dugas2015} (where the idea of projecting out exponential decays was first proposed for a single qubit protocol), can be regarded as special cases of our method.\\

\section*{Results}
In this section we present the main result of this paper: the character randomized benchmarking protocol, which leverages techniques from character theory~\cite{Fulton2004} to isolate the exponential decay channels in \cref{eq:rand_bench_fit_gen}. One can then fit these exponential decays one at a time, obtaining the quality parameters $f_\lambda$. We emphasize that the data generated by character randomized benchmarking can always be fitted to a single exponential, even if the gateset being benchmarked is not the Clifford group. Moreover our method retains its validity in the presence of leakage, which also causes deviations from single exponential behavior for standard randomized benchmarking~\cite{wallman2018randomized} (even when the gateset is the Clifford group).\\

For the rest of the paper we will use the Pauli Transfer Matrix (PTM) representation of quantum channels~\footnote{This representation is also sometimes called the Liouville representation or affine representation of quantum channels.~\cite{Wallman2014,Wolf2012}}. Key to this representation is the realization that the set of normalized non-identity Pauli matrices $\bf{\sigma}_q$ on $q$ qubits, together with the normalized identity $\sigma_0:= 2^{-q/2}\id$ forms an orthonormal basis (with respect to the trace inner product) of the Hilbert space of Hermitian matrices of dimension $2^q$. Density matrices $\rho$ and POVM elements $Q$ can then be seen as vectors and co-vectors expressed in the basis $\{\sigma_0\}\cup \bf{\sigma}_q$, denoted $\ket{\rho}$ and $\bra{Q}$ respectively. Quantum channels $\mc{E}$~\cite{Chuang1997} are then matrices (we will denote a channel and its PTM representation by the same letter) and we have $\mc{E}\ket{\rho}= \ket{\mc{E}(\rho)}$. Composition of channels $\mc{E},\mc{F}$ corresponds to multiplication of their PTM representations, that is $\ket{\mc{E}\circ \mc{F}(\rho)} = \mc{E}\mc{F}\ket{\rho}$. Moreover we can write expectation values as bra-ket inner products, i.e. $\bra{Q}\mc{E}\ket{\rho} = \tr(Q\mc{E}(\rho))$. The action of a unitary $G$ on a matrix $\rho$ is denoted $\mc{G}$, i.e. $\mc{G}\ket{\rho} = \ket{G\rho G\ct}$ and we denote its noisy implementation by $\widetilde{\mc{G}}$. For a more expansive review of the PTM representation, see section I.2 in the Supplementary Methods.\\

We will, for ease of presentation, also assume \emph{gate-independent noise}. This means we assume the existence of a CPTP map $\mc{E}$ such that $\tilde{\mc{G}} = \mc{E}\mc{G}$ for all $G\in\gr{G}$. We however emphasize that our protocol remains functional even in the presence of gate-dependent noise. We provide a formal proof of this, generalizing the modern treatment of standard randomized benchmarking with gate-dependent noise~\cite{wallman2018randomized}, in the Methods section. \\

\noindent{\bf Standard randomized benchmarking.}
Let's first briefly recall the ideas behind standard randomized benchmarking.
Subject to the assumption of gate-independent noise, the average survival probability $p_m$ of the standard randomized benchmarking procedure over a gateset $\gr{G}$ (with input state $\rho$ and measurement POVM $\{Q, \id - Q\}$) with sequence length $m$ can be written as~\cite{Magesan2012a}:
\begin{equation}\label{eq:rand_bench_av}
p_m = \bra{Q}\left(\avg_{G\in \gr{G}}\mc{G}\ct\mc{E}\mc{G}\right)^m\ket{\rho}.
\end{equation}
where $\md{E}_{G\in \gr{G}}$ denotes the uniform average over $\gr{G}$.
The key insight to randomized benchmarking is that $\mc{G}$ is a \emph{representation} (for a review of representation theory see section I.1 in the Supplementary Methods) of $G\in \gr{G}$. This representation will not be irreducible but will rather decompose into irreducible subrepresentations, that is $\mc{G} = \bigoplus_{\lambda\in R_\gr{G}}\phi_\lambda(G)$ where $R_\gr{G}$ is an index set and $\phi_\lambda$ are irreducible representations of $\gr{G}$ which we will assume to all be mutually inequivalent. Using Schur's lemma, a fundamental result in representation theory, we can write \cref{eq:rand_bench_av} as
\begin{equation}\label{eq:rand_bench_av_proj}
p_m = \sum_{\lambda} \bra{Q}  \mc{P}_\lambda\ket{\rho}f_\lambda^m
\end{equation}
where $\mc{P}_\lambda$ is the orthogonal projector onto the support of $\phi_\lambda$ (note that this is a superoperator) and $f_\lambda:= \tr(\mc{P}_\lambda \mc{E})/\tr(\mc{P}_\lambda)$ is the quality parameter associated to the representation $\phi_\lambda$ (note that the trace is taken over superoperators). This reproduces \cref{eq:rand_bench_fit_gen}. A formal proof of \cref{eq:rand_bench_av_proj} can be found in the Supplementary Methods and in~\cite{francca2018approximate}. The average fidelity of the gateset $\gr{G}$ can then be related to the parameters $f_\lambda$ as
\begin{equation}\label{eq:fid_gateset}
F_{\mathrm{avg}} = \frac{2^{-q}\sum_{\lambda\in R_{\gr{G}}} \tr(\mc{P}_\lambda)f_\lambda}{2^q+1}.
\end{equation}
Note again that $R_\gr{G}$ includes the trivial subrepresentation carried by $\ket{\id}$, so when $\mc{E}$ is a CPTP map there is a $\lambda\in R_{\gr{G}}$ for which $f_\lambda=1$.
See lemma's $4$ and $5$ in the Supplementary Methods for a proof of \cref{eq:fid_gateset}\\

\noindent{\bf Character randomized benchmarking.}
Now we present our new method called character randomized benchmarking. For this we make use of concepts from the character theory of representations~\cite{Fulton2004}.
Associated to any representation $\hat{\phi}$ of a group $\gr{\hat{G}}$ is a character function $\chi_{\hat{\phi}}:\gr{\hat{G}}\to \md{R}$, from the group to the real numbers~\footnote{Generally the character function is a map to the complex numbers, but in our case it is enough to only consider real representations.}. Associated to this character function is the following projection formula~\cite{Fulton2004}:
\begin{equation}\label{eq:proj_formula}
\avg_{\hat{G}\in \gr{\hat{G}}}\;\chi_{\hat{\phi}}(\hat{G})\mc{\hat{G}} = \frac{1}{|\hat{\phi}|}\mc{P}_{\hat{\phi}},
\end{equation}
where $\mc{P}_{\hat{\phi}}$ is the projector onto the support of all subrepresentations of $\mc{\hat{G}}$ equivalent to $\hat{\phi}$ and $|\hat{\phi}|$ is the dimension of the representation $\hat{\phi}$. We will leverage this formula to adapt the randomized benchmarking procedure in a way that singles out a particular exponential decay $f_\lambda^m$ in \cref{eq:rand_bench_fit_gen}.\\

We begin by choosing a group $\gr{G}$. We will call this group the `benchmarking group' going forward and it is for this group/gateset that we will estimate the average fidelity. In general we will have that $\mc{G} = \bigoplus_{\lambda\in R_\gr{G}}\phi_\lambda(G)$ where $R_\gr{G}$ is an index set and $\phi_\lambda$ are irreducible representations of $\gr{G}$ which we will assume to all be mutually inequivalent~\footnote{It is straightforward to extend character randomized benchmarking to also cover the presence of equivalent irreducible subrepresentation. However do not make this extension explicit here in the interest of simplicity}. Now fix a $\lambda'\in R_{\gr{G}}$. $f_{\lambda'}$ is the quality parameter associated to a specific subrepresentation $\phi_{\lambda'}$ of $\mc{G}$. Next consider a group $\hat{\gr{G}}\subset\gr{G}$ such that the PTM representation $\mc{\hat{G}}$ has a subrepresentation ${\hat{\phi}}$, with character function $\chi_{\hat{\phi}}$, that has support inside the representation $\phi_{\lambda'}$ of $\gr{G}$, i.e. $\mc{P}_{\hat{\phi}} \subset \mc{P}_{\lambda'}$ where $\mc{P}_{\lambda'}$ is again the projector onto the support of $\phi_{\lambda'}$. We will call this group $\gr{\hat{G}}$ the character group. Note that such a pair $\gr{\hat{G}},{\hat{\phi}}$ always exists; we can always choose $\gr{\hat{G}}=\gr{G}$ and ${\hat{\phi}}=\phi_{\lambda'}$. However other natural choices often exist, as we shall see when discussing examples of character randomized benchmarking. The idea behind the character randomized benchmarking protocol, described in \cref{char_rand_bench_box}, is now to effectively construct \cref{eq:proj_formula} by introducing the application of an extra gate $\hat{G}$ drawn at random from the character group $\gr{\hat{G}}$ into the standard randomized benchmarking protocol. In practice this gate will not be actively applied but must be compiled into the gate sequence following it, thus not resulting in extra noise (this holds even in the case of gate-dependent noise, see Methods).
\begin{figure}
\begin{framed}
\begin{enumerate}[leftmargin=*]
	\setlength\itemsep{-0.2em}
	\item Choose a state $\rho$ and a two-component POVM $\{Q, \id -Q\}$ such that $\tr(Q\mc{P}_{\hat{\phi}}(\rho))$ is large.
	\item Sample $ \vec{G}  =G_1,\ldots, G_m$ uniformly at random from $\gr{G}$ 
	\item Sample $\hat{G}$ uniformly at random from $\hat{\gr{G}}$ 
	\item Prepare the state $\rho$ and apply the gates $(G_1\hat{G}),G_2,\ldots G_m$
	\item Compute the inverse $G_{\mathrm{inv}} = (G_m\cdots G_1)\ct$ and apply it (note that $\hat{G}$ is not inverted)
	\item Estimate the weighted `survival probability' $k^{\hat{\lambda}'}_m(\vec{G},\hat{G})~=~|\hat{\phi}|\chi_{{\hat{\phi}}}(\hat{G})\bra{Q}\mc{\widetilde{G}}_{\mathrm{inv}}\mc{\widetilde{G}}_m\cdots \widetilde{(\mc{G}_1\mc{\hat{G}})}\ket{\rho}$
	\item Repeat for sufficient $\hat{G}\in \hat{\gr{G}}$ to estimate the average $k^{{\lambda}'}_m(\vec{G}) = \md{E}_{\hat{G}}(k^{{\lambda}'}_m(\vec{G},\hat{G}))$
	\item Repeat for sufficient $\vec{G}$ to estimate the average $k_m^{\lambda'} = \md{E}_{\vec{G}}(k^{{\lambda}'}_m(\vec{G}))$
	\item Repeat for sufficient different $m$ to fit to the exponential function $Af_{\lambda'}^m$ to obtain $f_{\lambda'}$
\end{enumerate}
\end{framed}
\caption{
{\bf The character randomized benchmarking protocol} Note the inclusion of the gate $\hat{G}$ and the average over the character function $\chi_{\hat{\phi}}$, which form the key ideas behind character randomized benchmarking. Note also that this extra gate $\hat{G}$ is compiled into the sequence of gates $(G_1,\ldots,G_m)$ and thus does not result in extra noise.}\label{char_rand_bench_box}
\end{figure}
This extra gate $\hat{G}\in \gr{\hat{G}}$ is not included when computing the global inverse $G_{\mathrm{inv}} = (G_1 \ldots G_m)\ct$.  The average over the elements of $\gr{\hat{G}}$ is also weighted by the character function $\chi_{\hat{\phi}}$ associated to the representation $\hat{\phi}$ of $\gr{\hat{G}}$.
Similar to \cref{eq:rand_bench_av} we can rewrite the uniform average over all $\vec{G}\in \gr{G}^{\times m}$ and $\hat{G}\in \gr{\hat{G}}$ as 
\begin{equation*}
k_m^{\lambda'} =|\hat{\phi}| \bra{Q}\!\!\left[\avg_{G\in \gr{G}}\mc{G}\ct \mc{E}\mc{G}\right]^m \!\! \!\!\avg_{\hat{G}\in \hat{\gr{G}}}\chi_{\hat{\phi}}(\hat{G})\mc{\hat{G}}\ket{\rho}.
\end{equation*}
Using the character projection formula (\cref{eq:proj_formula}), the linearity of quantum mechanics, and the standard randomized benchmarking representation theory formula (\cref{eq:rand_bench_av_proj}) we can write this as
\begin{equation}\label{eq:char_decay}
k_m^{\lambda'} = \sum_{\lambda\in R_\gr{G}} \bra{Q}\mc{P}_\lambda \mc{P}_{\hat{\phi}}\ket{\rho} f^m_\lambda = \bra{Q}\mc{P}_{\hat{\phi}}\ket{\rho}f^m_{\lambda'}
\end{equation}
since we have chosen $\gr{\hat{G}}$ and $\hat{\phi}$ such that $\mc{P}_{\hat{\phi}}\subset \mc{P}_{\lambda'}$. This means the character randomized benchmarking protocol isolates the exponential decay associated to the quality parameter $f_{\lambda'}$ independent of state preparation and measurement. We can now extract $f_{\lambda'}$ by fitting the data-points $k^{\lambda'}_m$ to a single exponential of the form $Af_{\lambda'}^m$. Note that this remains true even if $\mc{E}$ is not trace-preserving, i.e. the implemented gates experience leakage.
Repeating this procedure for all $\lambda'\in R_\gr{G}$ (choosing representations $\hat{\phi}$ of $\gr{\hat{G}}$ such that $\mc{P}_{\hat{\phi}}\subset\mc{P}_{\lambda'}$) we can reliably estimate all quality parameters $f_\lambda$ associated with randomized benchmarking over the group $\gr{G}$. Once we have estimated all these parameters we can use \cref{eq:fid_gateset} to obtain the average fidelity of the gateset $\gr{G}$.\\

\section*{Discussion}
We will now discuss several examples of randomized benchmarking experiments where the character randomized benchmarking approach is beneficial. The first example, benchmarking $T$-gates, is taken from the literature~\cite{Cross_2016} while the second one, performing interleaved benchmarking on a $2$-qubit gate using only single qubit gates a reference, is a new protocol. We have also implemented this last protocol to characterize a CPHASE gate between spin qubits in Si$\backslash$SiGe quantum dots, see~\cite{xue2018benchmarking}.\\

\noindent {\bf Benchmarking $T$-gates.}
The most common universal gateset considered in the literature is the Clifford$+T$ gateset~\cite{Nielsen2011}. The average fidelity of the Clifford gates can be extracted using standard randomized benchmarking over the Clifford group, but to extract the average fidelity of the $T$ gate a different approach is needed. Moreover one would like to characterize this gate in the context of larger circuits, meaning that we must find a family of multi-qubit groups that contains the $T$ gate. One choice is to perform randomized benchmarking over the group $\gr{T}_q$ generated by the $\CNOT$ gate between all pairs of qubits (in both directions), Pauli $X$  on all qubits and $T$ gates on all qubits (another choice would be to use dihedral randomized benchmarking~\cite{Dugas2015} but this is limited to single qubit systems, or to use the interleaved approach proposed in \cite{harper2017estimating}). This group is an example of a $\CNOT$-dihedral group and its use for randomized benchmarking was investigated in~\cite{Cross_2016}. There it was derived that the PTM representation of the group $\gr{T}_q$ decomposes into $3$ irreducible subrepresentations $\phi_1,\phi_2,\phi_3$ with associated quality parameters $f_1,f_2,f_3$ and projectors
\begin{equation*}
\mc{P}_1 = \ket{\sigma_0}\!\bra{\sigma_0},\;\;\mc{P}_2 =\!\!\sum_{\sigma\in \mc{Z}_q} \ket{\sigma}\!\bra{\sigma},\;\;\mc{P}_3 =\!\!\!\!\!\sum_{\sigma\in{\bf\sigma}_q/\mc{Z}_q}\! \ket{\sigma}\!\bra{\sigma},
\end{equation*}
where $\sigma_0$ is the normalized identity, ${\bf\sigma}_q$ is the set of normalized Pauli matrices and $\mc{Z}_q$ is the subset of the normalized Pauli matrices composed only of tensor products of $Z$ and $\id$. Noting that $f_1 =1$ if the implemented gates $\mc{\widetilde{G}}$ are CPTP we must estimate $f_2$ and $f_3$ in order to estimate the average fidelity of $\gr{T}_q$. Using standard randomized benchmarking this would thus lead to a two-decay, four-parameter fitting problem, but using character randomized benchmarking we can fit $f_2$ and $f_3$ separately. Let's say we want to estimate $f_2$, associated to $\phi_2$, using character randomized benchmarking. In order to perform character randomized benchmarking we must first choose a character group $\gr{\hat{G}}$. A good choice for $\gr{\hat{G}}$ is in this case the Pauli group $\gr{P}_q$. Note that $\gr{P}_q\subset\gr{T}_q$ since $T^4 = Z$ the Pauli Z matrix.

Having chosen $\gr{\hat{G}} = \gr{P}_q$ we must also choose an irreducible subrepresentation $\hat{\phi}$ of the PTM representation of the Pauli group $\gr{P}_q$ such that $\mc{P}_{\hat{\phi}}\mc{P}_2 = \mc{P}_{\hat{\phi}}$.
As explained in detail in section V.I in the Supplementary Methods the PTM representation of the Pauli group has $2^q$ irreducible inequivalent subrepresentations of dimension one. These representations $\phi_\sigma$ are each associated to an element $\sigma\in \{\sigma_0\}\cup\bf{\sigma}_q$ of the Pauli basis. Concretely we have that the projector onto the support of $\phi_\sigma$ is given by $\mc{P}_\sigma = \ket{\sigma}\!\bra{\sigma}$. This means that, to satisfy $\mc{P}_{\hat{\phi}}\mc{P}_2 = \mc{P}_{\hat{\phi}}$ we have to choose $\hat{\phi} = \phi_\sigma$ with $\sigma \in \mc{Z}_q$. One could for example choose $\sigma$ proportional to $Z\tn{q}$.
 The character associated to the representation $\phi_\sigma$ is $\chi_{\sigma}(P) = (-1)^{\inp{P}{\sigma}}$ where $\inp{P}{\sigma}=1$ if and only if $P$ and $\sigma$ anti-commute and zero otherwise (we provided a proof of this fact in section V.1 of the Supplementary Methods). Hence the character randomized benchmarking experiment with benchmarking group $\gr{T}_q$, character group $\gr{P}_q$ and subrepresentation $\hat{\phi}=\phi_\sigma$ produces data that can be described by
\begin{equation}
k_m^2 = \bra{Q}\sigma\rangle\!\rangle\!\langle\!\langle\sigma\ket{\rho}f_2^m,
\end{equation}
allowing us to reliably extract the parameter $f_2$. We can perform a similar experiment to extract $f_3$, but we must instead choose $\sigma \in \bsq\backslash\mc{Z}$. A good choice would for instance be $\sigma$ proportional to $X\tn{q}$.\\
Having extracted $f_2$ and $f_3$ we can then use \cref{eq:fid_gateset} to obtain the average fidelity of the gateset $\gr{T}_q$ as~\cite{Cross_2016}:
\begin{equation}
F_{\mathrm{avg}} = \frac{2^q-1}{2^q}\left(1 - \frac{f_2 + 2^q f_3}{2^q+1}\right)
\end{equation}
Finally we would like to note that in order to get good signal one must choose $\rho$ and $Q$ appropriately. The correct choice is suggested by \cref{eq:char_decay}. For instance, if when estimating $f_2$ as above we choose $\sigma$ proportional to $Z\tn{q}$ we must then choose $Q = \frac{1}{2}(\id+Z\tn{2})$ and $\rho  = \frac{1}{d}(\id+Z\tn{2})$. This corresponds to the even parity eigenspace (in the computational basis).\\

\noindent{\bf 2-for-1 interleaved benchmarking.}\label{ssec:sim_bench}
The next example is a new protocol, which we call 2-for-1 interleaved randomized benchmarking. It is a way to perform interleaved randomized benchmarking~\cite{Magesan_2012_interleaved} of a $2$-qubit Clifford gate $C$ using only single qubit Clifford gates as reference gates. The advantages of this are (1) lower experimental requirements and (2) a higher reference gate fidelity relative to the interleaved gate fidelity allows for a tighter estimate of the average fidelity of the interleaved gate (assuming single qubit gates have higher fidelity than two qubit gates). This latter point is related to an oft overlooked drawback of interleaved randomized benchmarking, namely that it does not yield a direct estimate of the average fidelity $F(C)$ of the interleaved gate $C$ but only gives upper and lower bounds on this fidelity. These upper and lower bounds moreover depend~\cite{dugas2016efficiently,Magesan_2012_interleaved} on the fidelity of the reference gates and can be quite loose if the fidelity of the reference gates is low. To illustrate the advantages of this protocol we have performed a simulation comparing it to standard interleaved randomized benchmarking (details can be found in section V.2 in the Supplementary Methods). Following recent single qubit randomized benchmarking and Bell state tomography results in spin qubits in Si$\backslash$SiGe quantum dots~\cite{watson2018programmable,zajac2018resonantly,huang2018fidelity} we assumed single qubit gates to have a fidelity of $F_{\mathrm{avg}}^{(1)} = 0.987$ and two-qubit gates to have a fidelity of $F_{\mathrm{avg}}(C) = 0.898$. Using standard interleaved randomized benchmarking~\cite{Magesan_2012_interleaved} we can guarantee (using the optimal bounds of \cite{dugas2016efficiently}) that the fidelity of the interleaved gate is lower bounded by $F_{\mathrm{avg}}^{\mathrm{int}} \approx 0.62$ while using 2-for-1 interleaved randomized benchmarking we can guarantee that the fidelity of interleaved gate is lower bounded by $F_{\mathrm{avg}}(C) \approx 0.79$, a significant improvement that is moreover obtained by a protocol requiring less experimental resources. On top of this the 2-for-1 randomized benchmarking protocol provides strictly more information than simply the average fidelity, we can also extract a measure of correlation between the two qubits, as per~\cite{Gambetta_2012_sim}. In another paper~\cite{xue2018benchmarking} we have used this protocol to characterize a CPHASE gate between spin qubits in Si$\backslash$SiGe quantum dots.\\

An interleaved benchmarking experiment consists of two stages, a reference experiment and an interleaved experiment. The reference experiment for 2-for-1 interleaved randomized benchmarking consists of character randomized benchmarking using $2$ copies of the single-qubit Clifford group $\gr{G} = \gr{C}_1\tn{2}$ as the benchmarking group (this is also the group considered in simultaneous randomized benchmarking~\cite{Gambetta_2012_sim}). The PTM representation of $\gr{C}_1\tn{2}$decomposes into four irreducible subrepresentations and thus the fitting problem of a randomized benchmarking experiment over this group involves $4$ quality parameters $f_w$ indexed by $w=(w_1,w_2)\in \{0,1\}^{\times 2}$. The projectors onto the associated irreducible representations $\phi_w$ are
\begin{equation}\label{eq:2for1_reps}
\mc{P}_w =\sum_{\sigma\in {\bf\sigma}_w} \ket{\sigma}\!\bra{\sigma}
\end{equation}
where ${\bf \sigma}_w$ is the set of normalized $2$-qubit Pauli matrices that have non-identity Pauli matrices at the $i$'th tensor factor if and only if $w_i=1$. To perform character randomized benchmarking we choose as character group $\gr{\hat{G}} =\gr{P}_2$ the $2$-qubit Pauli group. For each $w\in\{0,1\}^{\times 2}$ we can isolate the parameter $f_w$ by correctly choosing a subrepresentation $\phi_\sigma$ of the PTM representation of $\gr{P}_2$. Recalling that $\mc{P}_\sigma = \ket{\sigma}\!\bra{\sigma}$ we can choose $\hat{\phi} = \phi_{\sigma}$ for $\sigma = (Z^w_1\otimes Z^w_2)/2$ to isolate the parameter $f_w$ for $w=(w_1,w_2)\in \{0,1\}^{\times 2}$. We give the character functions associated to these representation in section V.2 of the Supplementary Methods.
Once we have obtained all quality parameters $f_w$ we can compute the average reference fidelity $F_{\mathrm{ref}}$ using \cref{eq:fid_gateset}.

The interleaved experiment similarly consists of a character randomized benchmarking experiment using $\gr{G} = \gr{C}_1\tn{2}$ but for every sequence $\vec{G} = (G_1,\ldots,G_m)$ we apply the sequence $(G_1,C,G_2,\ldots,C,G_m)$ instead, where $C$ is a $2$-qubit interleaving gate (from the $2$-qubit Clifford group). Note that we must then also invert this sequence (with $C$) to the identity~\cite{Magesan_2012_interleaved}. Similarly choosing $\gr{\hat{G}}=\gr{P}_2$ we can again isolate the parameters $f_w$ and from these compute the `interleaved fidelity' $F_{\mathrm{int}}$. Using the method detailed in~\cite{dugas2016efficiently} we can then calculate upper and lower bounds on the average fidelity $F_{\mathrm{avg}}(C)$ of the gate $C$ from the reference fidelity $F_{\mathrm{ref}}$ and the interleaved fidelity $F_{\mathrm{int}}$. Note that it is not trivial that the interleaved experiment yields data that can be described by a single exponential decay, we will discuss this in greater detail in the methods section.\\

Finally we would like to note that the character benchmarking protocol can be used in many more scenarios than the ones outlined here. Character randomized benchmarking is versatile enough that when we want to perform randomized benchmarking we can consider first what group is formed by the native gates in our device and then use character benchmarking to extract gate fidelities from this group directly, as opposed to carefully compiling the Clifford group out of the native gates which would be required for standard randomized benchmarking. This advantage is especially pronounced when the native two-qubit gates are not part of the Clifford group, which is the case for e.g. the $\sqrt{\mathrm{SWAP}}$ gate~\cite{kalra2014robust,li2018crossbar}.

\section*{Methods}
In this section will discuss three things: (1) The statistical behavior and scalability of character randomized benchmarking, (2) the robustness of character randomized benchmarking against gate-dependent noise, and (3) the behavior of interleaved character randomized benchmarking, and in particular 2-for-1 interleaved benchmarking.\\

\noindent First we will consider whether the character randomized benchmarking protocol is efficiently scalable with respect to the number of qubits (like standard randomized benchmarking) and whether the character randomized benchmarking protocol remains practical when only a finite amount of data can be gathered (this last point is a sizable line of research for standard randomized benchmarking~\cite{Wallman2014,helsen2017multi,hincks2018bayesian,Granade2014}).\\

 \noindent{\bf Scalability of character randomized benchmarking.}

The resource cost (the number of experimental runs that must be performed to obtain an estimate of the average fidelity) of character randomized benchmarking can be split into two contributions: (1) The number of quality parameters $f_{\lambda}$ associated that must be estimated  (this is essentially set by $|R_{\gr{G}}|$, the number of irreducible subrepresentations of the PTM representation of the benchmarking group $\gr{G}$), and (2) the cost of estimating a single average $k^{\lambda'}_m$ for a fixed $\lambda'\in R_{\gr{G}}$ and sequence length $m$. \\

The first contribution implies that for scalable character randomized benchmarking with (a uniform family of) groups $\gr{G}_q$ (w.r.t. the number of qubits $q$) the number of quality parameters (set by $|R_{\gr{G}}|$) must grow polynomially with $q$. This means that not all families of benchmarking groups are can be characterized by character randomized benchmarking in a scalable manner.\\

The second contribution, as can be seen in \cref{char_rand_bench_box}, further splits up into three components: (2a) the magnitude of $|\hat{\phi}|$, (2b) the number of random sequences $\vec{G}$ needed to estimate $k^{\lambda'}_m$ (given access to $k^{\lambda'}_m(\vec{G})$ ) and (2c) the number of samples needed to estimate $k^{\lambda'}_m(\vec{G})$ for a fixed sequence. We will now argue that the resource cost of all three components are essentially set by the magnitude of $|\hat{\phi}|$. Thus if $|\hat{\phi}|$ grows polynomially with the number of qubits then the entire resource cost does so as well. Hence a sufficient condition for scalable character randomized benchmarking is that one chooses a family of benchmarking groups where $|R_{\gr{G}}|$ grows polynomially in $q$ and character groups such that for the relevant subrepresentations $|\hat{\phi}|$ the dimension grows polynomially in $q$. \\

We begin by arguing (2c):The character-weighted average over the group $\gr{\hat{G}}$ for a single sequence $\vec{G}$: $k_m^{\lambda'}(\vec{G})$, involves an average over $|\gr{\hat{G}}|$ elements (which will generally scale exponentially in $q$), but can be efficiently estimated by not estimating each character-weighted expectation value $k_m^{{\lambda}'}(\vec{G},\hat{G})$ individually but rather estimate $k_m^{{\lambda}'}(\vec{G})$ directly by the following procedure

\begin{enumerate}[leftmargin=*]
	\setlength\itemsep{-0em}
	\item Sample $\hat{G}\in \gr{\hat{G}}$ uniformly at random 
	\item Prepare the state $\mc{G}_{\mathrm{inv}}\mc{G}_m\cdots \mc{G}_1\mc{\hat{G}}\ket{\rho}\vspace{0.5mm}$ and measure it once obtaining a result $b(\hat{G}) \in\{0,1\}$
	\item Compute $x(\hat{G})~\!\!:=~\!\!\chi_{\hat{\phi}}(\hat{G})|\hat{\phi}|b(\hat{G})~\!\!\in~\!\!\{0,\!\chi_{\hat{\phi}}(\hat{G})|\hat{\phi}|\} \vspace{1mm}$
	\item Repeat sufficiently many times and compute the empirical average of $x(\hat{G})$$\vspace{-0.5em}$
\end{enumerate}
Through the above procedure we are directly sampling from a bounded probability distribution with mean $k^{{\lambda}'}_m(\vec{G})$ that takes values in the interval $[-{\chi}_{\hat{\phi}}^*,{\chi}_{\hat{\phi}}^*]$ where ${\chi}_{\hat{\phi}}^*$ is the largest absolute value of the character function ${\chi}_{\hat{\phi}}$. Since the maximal absolute value of the character function is bounded by the dimension of the associated representation~\cite{Fulton2004}, this procedure will be efficient as long as $|\hat{\phi}|$ is not too big. 

For the examples given in the discussion section (with the character group being the Pauli group) the maximal character value is $1$. Using standard statistical techniques~\cite{Hoeffding1963} we can give e.g. a $99\%$ confidence interval of size $0.02$ around $k^{{\lambda}'}_m(\vec{G})$ by repeating the above procedure $1769$ times, which is within an order of magnitude of current experimental practice for confidence intervals around regular expectation values and moreover independent of the number of qubits $q$. See section VI in the Supplementary Methods for more details on this.

We now consider (2b): From the considerations above we know that $k_m^{\lambda'}(\vec{G})$ is the mean of a set of random variables and thus itself a random variable, taking possible values in the interval $[-{\chi}_{\hat{\phi}}^*,{\chi}_{\hat{\phi}}^*]$. Hence by the same reasoning as above we see that $k_m^{\lambda'}$, as the mean of a distribution (induced by the uniform distribution of sequences $\vec{G}$) confided to the interval $[-{\chi}_{\hat{\phi}}^*,{\chi}_{\hat{\phi}}^*]$ can be estimated using an amount of resources polynomially bounded in $|\hat{\phi}|$. We would like to note however that this estimate is probably overly pessimistic in light of recent results for standard randomized benchmarking on the Clifford group~\cite{Wallman2014,helsen2017multi} where it was shown that the average $k_m^{\lambda'}$ over sequences $\vec{G}\in \gr{G}^{\times m}$ can be estimated with high precision and high confidence using only a few hundred sequences. These results depend on the representation theoretic structure of the Clifford group but we suspect that it is possible to generalize these results at least partially to other families of benchmarking groups. Moreover any such result can be straightforwardly adapted to also hold for character randomized benchmarking. Actually making such estimates for other families groups is however an open problem, both for standard and character randomized benchmarking.\\

To summarize, the scalability of character randomized benchmarking depends on the properties of the families of benchmarking and character groups chosen. One should choose the benchmarking groups such that the number of exponential decays does not grow too rapidly with the number of qubits, and one should choose the character group such that the dimension of the representation being projected on does not grow too rapidly with the number of qubits.

\noindent{\bf Gate-dependent noise}
Thus far we have developed the theory of character randomized benchmarking under the assumption of gate-independent noise. This is is not a very realistic assumption. Here we will generalize our framework to include gate-dependent noise. In particular we will deal with the so called `non-Markovian' noise model. This noise model is formally specified by the existence of a function $\Phi:\gr{G}\rightarrow\mc{S}_{2^q}$ which assigns to each element $G$ of the group $\gr{G}$ a quantum channel $\Phi(G) = \mc{E}_G$. Note that this model is not the most general, it does not take into account the possibility of time dependent effects or memory effects during the experiment. It is however much more general and realistic than the gate-independent noise model.
 In this section we will prove two things:
\begin{enumerate}
	\item A character randomized benchmarking experiment always yields data that can be fitted to a single exponential decay up to a small and exponentially decreasing corrective term.
	\item The decay rates yielded by a character randomized benchmarking experiment can be related to the average fidelity (to the identity) of the noise in between gates, averaged over all gates.
\end{enumerate}
Both of these statements, and their proofs, are straightforward generalizations of the work of Wallman~\cite{wallman2018randomized} which dealt with standard randomized benchmarking. We will see that his conclusion, that randomized benchmarking measures the average fidelity of noise in between quantum gates up to a small correction, generalizes to the character benchmarking case. We begin with a technical theorem, which generalizes \cite[theorem 2]{wallman2018randomized} to twirls over arbitrary groups (with multiplicity-free PTM representations). 
\begin{theorem}\label{thm:eig_ops}
Let $\gr{G}$ be a group such that its PTM representation $\mc{G} = \bigoplus_{\lambda\in R_\gr{G}}\phi_\lambda(G)$ is multiplicity-free. Denote for all $\lambda$ by $f_\lambda$ the largest eigenvalue of the operator $\md{E}_{G\in\gr{G}}(\tilde{\mc{G}}\otimes \phi_\lambda(G))$ where $\tilde{\mc{G}}$ is the CPTP implementation of $G\in \gr{G}$. There exist Hermicity-preserving linear superoperators $\mc{L},\mc{R}$ such that 
\begin{align}
\avg_{G\in \gr{G}}(\tilde{\mc{G}}\mc{L}\mc{G}\ct) = \mc{L}\mc{D}_{\gr{G}},\label{eq:L_eig}\\
\avg_{G\in \gr{G}}(\mc{G}\ct\mc{R}\tilde{\mc{G}}) = \mc{D}_{\gr{G}}\mc{R},\label{eq:R_eig}\\
\avg_{G\in \gr{G}}(\mc{G}\mc{R}\mc{L}\mc{G}\ct) = \mc{D}_{\gr{G}},\label{eq:RL_eig}
\end{align}
where $\mc{D}_{\gr{G}}$ is defined as
\begin{equation}
\mc{D}_{\gr{G}} = \sum_{\lambda}f_\lambda \mc{P}_{\lambda},
\end{equation}
with $\mc{P}_{\lambda}$ the projector onto the representation $\phi_\lambda$ for all $\lambda\in R_\gr{G}$.
\end{theorem}

\begin{proof}
Using the definition of $\mc{G}$ and $\mc{D}_{\gr{G}}$ we can rewrite \cref{eq:L_eig} as 
\begin{align}
\sum_\lambda \avg_{G\in \gr{G}}(\tilde{\mc{G}}(\mc{L}\mc{P}_\lambda) \phi_\lambda(G)\ct) =\sum_\lambda f_\lambda \mc{L}\mc{P}_{\lambda}.
\end{align}
This means that, without loss of generality, we can take $\mc{L}$ to be of the form
\begin{equation}\label{eq:L_lambda_def}
\mc{L} =\sum_\lambda \mc{L}_\lambda ,\;\;\;\;\;\;\;\;\;\;\;\;\mc{L}_\lambda \mc{P}_{\lambda'} = \delta_{\lambda\lambda'} \mc{L}_\lambda ,\;\;\;\forall \lambda'.
\end{equation}
Similarly we can take $\mc{R}$ to be
\begin{equation}\label{eq:R_lambda_def}
\mc{R} =\sum_\lambda \mc{R}_\lambda ,\;\;\;\;\;\;\;\;\;\;\;\;\mc{P}_{\lambda'}\mc{R}_\lambda = \delta_{\lambda\lambda'} \mc{R}_\lambda ,\;\;\;\forall \lambda'.
\end{equation}
This means \cref{eq:R_eig,eq:L_eig} decompose into independent pairs of equations for each $\lambda$:
\begin{align}
\avg_{G\in \gr{G}}(\tilde{\mc{G}}\mc{L}_\lambda\phi_{\lambda}(G)\ct) = f_\lambda\mc{L}_\lambda\label{eq:L_eig_lambda}\\
\avg_{G\in \gr{G}}(\phi_\lambda(G)\ct\mc{R}\tilde{\mc{G}}) = f_\lambda \mc{R}_\lambda.\label{eq:R_eig_lambda}
\end{align}
Next we use the vectorization operator $\mathrm{vec}:\mathrm{M}_{2^{2q}} \rightarrow \md{R}^{2^{4q}}$ mapping the PTM representations of superoperators to vectors of length $\md{R}^{2^{4q}}$. This operator has the property that for all $A,B,C\in \mathrm{M}_{2^{2q}}$ we have
\begin{equation}
\mathrm{vec}(ABC) = A\otimes C^T \mathrm{vec}(B)
\end{equation}
where $C^T$ is the transpose of $C$. Applying this to the equations \cref{eq:L_eig_lambda,eq:R_eig_lambda} and noting that $\mc{G}\ct = \mc{G}^T$ since $\mc{G}$ is a real matrix we get the eigenvalue problems equivalent to \cref{eq:L_eig_lambda,eq:R_eig_lambda},
\begin{align}
\avg_{G\in \gr{G}}(\tilde{\mc{G}}\otimes \phi_{\lambda}(G) )\mathrm{vec}{(\mc{L}_\lambda)} = f_\lambda \mathrm{vec}{(\mc{L}_\lambda)}\\
\avg_{G\in \gr{G}}(\tilde{\mc{G}}\otimes \phi_{\lambda}(G) )^T\mathrm{vec}{(\mc{R}_\lambda)} = f_\lambda \mathrm{vec}{(\mc{R}_\lambda)}.
\end{align}
Since we have defined $f_\lambda$ to be the largest eigenvalue of $\md{E}_{G\in \gr{G}}(\tilde{\mc{G}}\otimes \phi_{\lambda}(G) )$ (and equivalently of $\md{E}_{G\in \gr{G}}(\tilde{\mc{G}}\otimes \phi_{\lambda}(G) )^T$) we can choose $\mathrm{vec}(\mc{L})$ and $\mathrm{vec}(\mc{R})$ to be the left and right eigenvectors respectively of $\md{E}_{G\in \gr{G}}(\tilde{\mc{G}}\otimes \phi_{\lambda}(G) )$ associated to $f_\lambda$. Inverting the vectorization we obtain solutions to the equations \cref{eq:L_eig_lambda,eq:R_eig_lambda} and hence also \cref{eq:L_eig,eq:R_eig}. To see that this solution also satisfies \cref{eq:RL_eig} we note first that $\md{E}_{G\in \gr{G}}(\mc{G}\mc{R}_\lambda\mc{L}_\lambda\mc{G}\ct)$ is proportional to $P_\lambda$ for any  $\mc{R}_\lambda,\mc{L}_\lambda$ satisfying \cref{eq:L_lambda_def,eq:R_lambda_def} (by Schur's lemma). Since the eigenvectors of $\md{E}_{G\in \gr{G}}(\tilde{\mc{G}}\otimes \phi_{\lambda}(G) )$ are only defined up to a constant we can for every $\lambda$ choose proportionality constants such that $\md{E}_{G\in \gr{G}}(\mc{G}\mc{R}_\lambda\mc{L}_\lambda\mc{G}\ct)=f_\lambda P_\lambda$ and thus that \cref{eq:RL_eig} is satisfied.

\end{proof}
Next we prove that if we perform a character randomized benchmarking experiment with benchmarking group $\gr{G}$, character group $\hat{G}$ and subrepresentations $\hat{\phi}\subset \phi_{\lambda'}$ for some $\lambda'\in R_{\gr{G}}$, the observed data can always be fitted (up to an exponentially small correction) to a single exponential decay. The decay rate of $f_{\lambda'}$ associated to this experiment will be the largest eigenvalue of the operator $\md{E}_{G\in \gr{G}}(\tilde{\mc{G}}\otimes \phi_{\lambda'}(G) )$ mentioned in the theorem above. Later we will give an operational interpretation of this number. We begin by defining, for all $G\in \gr{G}$ a superoperator $\Delta_G$ which captures the `gate-dependence' of the noise implementation of $\mc{G}$,
\begin{equation}
\Delta_G := \tilde{\mc{G} }- \mc{L}\mc{G}\mc{R},
\end{equation}
where $\mc{R},\mc{L}$ are defined as in \cref{thm:eig_ops}. Using this expansion we have the following theorem, which generalizes \cite[theorem 4]{wallman2018randomized} to character randomized benchmarking over arbitrary finite groups with multiplicity-free PTM representation.

\begin{theorem}
Let $\gr{G}$ be a group such that its PTM representation $\mc{G} = \bigoplus_{\lambda\in R_\gr{G}}\phi_\lambda(G)$ is multiplicity-free.  Consider the outcome of a character randomized benchmarking experiment with benchmarking group $\gr{G}$, character group $\hat{G}$, subrepresentations $\hat{\phi}\subset \phi_{\lambda'}$ for some $\lambda'\in R_{\gr{G}}$, and set of sequence lengths $\md{M}$. That is, consider the real number
\begin{equation}
k_m^{\lambda'}=\avg_{G\in\gr{G}}\avg_{\hat{G}\in \gr{\hat{G}}} \chi_{\hat{\phi}}(\hat{G})|{\hat{\phi}}|\bra{Q}\tilde{\mc{G}}_{\mathrm{inv}}\tilde{\mc{G}}_m\cdots \widetilde{\mc{G}_1\mc{\hat{G}}}\ket{\rho}
\end{equation}
for some input state $\rho$ and output POVM $\{Q, \id -Q\}$ and $m\in \md{M}$ . This probability can be fitted to an exponential of the form
\begin{equation}
k_m^{\lambda'} =_{\mathrm{fit}} Af_{\lambda'}^m + \varepsilon_m,
\end{equation}
where $A$ is a fitting parameter, $f_\lambda$ is the largest eigenvalue of the operator $\md{E}_{G\in \gr{G}}(\tilde{\mc{G}}\otimes \phi_\lambda(G))$ and $\varepsilon_m \leq \delta_1\delta_2^m$ with 
\begin{align}
\delta_1 = |{\hat{\phi}}|\max_{\hat{G}\in \gr{\hat{G}}}|\chi_{\hat{\phi}}(\hat{G})|\max_{G\in \gr{G}} \norm{\Delta_G}_\diamond,\\
\delta_2 =\md{E}_{G \in \gr{G}}\norm{\Delta_G}_{\diamond},
\end{align}
where $\norm{\cdot}_\diamond$ is the diamond norm on superoperators~\cite{watrous2004notes}.
\end{theorem}
\begin{proof}
We begin by expanding $\widetilde{\mc{G}_1\mc{\hat{G}}} = \mc{L}\mc{G}_1\mc{\hat{G}}\mc{R} + \Delta_{G_1\hat{G}}$. This gives us
\begin{align}\label{eq:two_terms}
k_m^{\lambda'}&=\avg_{\substack{\hat{G}\in \gr{\hat{G}}\\G_1,\ldots,G_m\in\gr{G}}} \!\!\!\!\!\chi_{\hat{\phi}}(\hat{G})|{\hat{\phi}}|\bra{Q}\tilde{\mc{G}}_{\mathrm{inv}}\tilde{\mc{G}}_m\cdots \mc{L}\mc{G}_1\mc{\hat{G}}\mc{R}\ket{\rho} \\&\hspace{4em}+ \chi_{\hat{\phi}}(\hat{G})|{\hat{\phi}}|\bra{Q}\tilde{\mc{G}}_{\mathrm{inv}}\tilde{\mc{G}}_m\cdots \Delta_{G_1\hat{G}}\ket{\rho}.
\end{align}
We now analyze the first term in \cref{eq:two_terms}. Using the character projection formula, the fact that $\mc{G}_1 = (\mc{G}_{inv}\mc{G}_{m}\ldots\mc{G}_2)\ct$ and \cref{eq:L_eig} from \cref{thm:eig_ops} we get
\begin{align}
&\!\!\!\!\!\!\!\avg_{\substack{\hat{G}\in \gr{\hat{G}}\\G_1,\ldots,G_m\in\gr{G}}} \!\!\!\!\!\chi_{\hat{\phi}}(\hat{G})|{\hat{\phi}}|\bra{Q}\tilde{\mc{G}}_{\mathrm{inv}}\tilde{\mc{G}}_m\cdots \mc{L}\mc{G}_1\mc{\hat{G}}\mc{R}\ket{\rho}\notag\\&\hspace{0em}=\!\!\!\! \avg_{G_1,\ldots,G_m\in\gr{G}} \!\!\!\!\!\bra{Q}\tilde{\mc{G}}_{\mathrm{inv}}\tilde{\mc{G}}_m\cdots\mc{\tilde{G}}_2\mc{L}\mc{G}_2\ct \ldots\mc{G}_{\mathrm{inv}}\ct \mc{P}_{\hat{\phi}}\mc{R}\ket{\rho}\\
&\hspace{0em}=\!\!\!\!\avg_{G_3,\ldots,G_m\in\gr{G}} \!\!\!\!\!\bra{Q}\tilde{\mc{G}}_{\mathrm{inv}}\tilde{\mc{G}}_m\cdots\mc{\tilde{G}}_3\mc{L}\mc{D}_\gr{G}\mc{G}_3\ct \ldots\mc{G}_{\mathrm{inv}}\ct \mc{P}_{\hat{\phi}}\mc{R}\ket{\rho}\\
&\hspace{0em}= \bra{Q}\mc{L}\mc{D}_\gr{G}^m \mc{P}_{\hat{\phi}}\mc{R}\ket{\rho}\\
&\hspace{0em}= f_{\lambda'}^m \bra{Q}\mc{L} \mc{P}_{\hat{\phi}}\mc{R}\ket{\rho}
\end{align}
where we used that $\mc{D}_\gr{G}$ commutes with $\mc{G}$ for all $G\in \gr{G}$ and the fact that $\mc{D}_\gr{G}\mc{P}_{\hat{\phi}} = f_{\lambda'} \mc{P}_{\hat{\phi}}$. Next we consider the second term in \cref{eq:two_terms}. For this we first need to prove a technical statement. We make the following calculation for all $j\geq 2$ and $\hat{G}\in \gr{\hat{G}}$:
\begin{align}
&\avg_{G_1,\ldots,G_m\in\gr{G}}\tilde{\mc{G}}_{\mathrm{inv}}\tilde{\mc{G}}_m\cdots\mc{\tilde{G}}_{j+1}\mc{L}\mc{G}_j\mc{R}\Delta_{G_{j-1}}\ldots \Delta_{G_1\hat{G}}\\
&= \avg_{G_1,\ldots,G_m\in\gr{G}}\tilde{\mc{G}}_{\mathrm{inv}}\tilde{\mc{G}}_m\cdots\mc{\tilde{G}}_{j+1}\mc{L}\mc{G}_{j+1}\ct\ldots \mc{G}_m\ct\notag\\
&\hspace{2em}\times \mc{G}_{\mathrm{inv}}\mc{G}_1\ct\ldots\mc{G}_{j-1}\ct\mc{R}\Delta_{G_{j-1}}\ldots \Delta_{G_1\hat{G}}\\
&=\avg_{G_1,\ldots,G_m\in\gr{G}}\!\!\!\!\!\tilde{\mc{G}}_{\mathrm{inv}}\tilde{\mc{G}}_m\cdots\mc{\tilde{G}}_{j+1}\mc{L}\mc{G}_{j+1}\ct\ldots \mc{G}_m\ct\notag\\
&\hspace{2em}\times \mc{G}_{\mathrm{inv}}\mc{G}_1\ct\ldots\mc{G}_{j-1}\ct\mc{R}(\mc{\tilde{G}}_{j-1} - \mc{L}\mc{G}_{j-1}\mc{R})\notag\\
&\hspace{13em}\times \Delta_{G_{j-2}}\ldots \Delta_{G_1\hat{G}}\\
&=\avg_{\substack{G_1,\ldots,G_{j-1},\\G_{j+1},\ldots G_m \in\gr{G}}}\!\!\!\!\!\tilde{\mc{G}}_{\mathrm{inv}}\tilde{\mc{G}}_m\cdots\mc{\tilde{G}}_{j+1}\mc{L}\mc{G}_{j+1}\ct\ldots \mc{G}_m\ct \notag\\
&\hspace{2em}\times\mc{G}_{\mathrm{inv}}\mc{G}_1\ct\ldots\mc{G}_{j-2}\ct(\mc{D}_{\gr{G}} - \mc{D}_\gr{G})\mc{R}\Delta_{G_{j-2}}\ldots \Delta_{G_1\hat{G}}\\
&=0
\end{align}
where we used the definition of $\Delta_{G_{j-1}}$, the fact that $G_{j-1} = (G_m\ldots G_{j+1})\ct G_{\mathrm{inv}}(G_1\ldots G_{j-1})\ct$ and \cref{eq:R_eig,eq:RL_eig}. We can apply this calculation to the second term of \cref{eq:two_terms} to get
\begin{align}
&\avg_{\substack{\hat{G}\in \gr{\hat{G}}\\G_1,\ldots,G_m\in\gr{G}}} \!\!\!\!\!\chi_{\hat{\phi}}(\hat{G})|{\hat{\phi}}|\bra{Q}\tilde{\mc{G}}_{\mathrm{inv}}\tilde{\mc{G}}_m\cdots\mc{\tilde{G}}_2 \Delta_{G_1\hat{G}}\ket{\rho} \\&=\!\!\!\! \avg_{\substack{\hat{G}\in \gr{\hat{G}}\\G_1,\ldots,G_m\in\gr{G}}} \!\!\!\!\!\chi_{\hat{\phi}}(\hat{G})|{\hat{\phi}}|\bra{Q}\tilde{\mc{G}}_{\mathrm{inv}}\tilde{\mc{G}}_m\cdots(\mc{L}\mc{G}_2\mc{R} + \Delta_{G_2})\notag\\&\hspace{14em}\vspace{5em}\times \Delta_{G_1\hat{G}}\ket{\rho}\\
&=\!\!\!\! \avg_{\substack{\hat{G}\in \gr{\hat{G}}\\G_1,\ldots,G_m\in\gr{G}}} \!\!\!\!\!\chi_{\hat{\phi}}(\hat{G})|{\hat{\phi}}|\bra{Q}\tilde{\mc{G}}_{\mathrm{inv}}\tilde{\mc{G}}_m\cdots\tilde{\mc{G}}_3 \Delta_{G_2} \Delta_{G_1\hat{G}}\ket{\rho}\\
&=\!\!\!\! \avg_{\substack{\hat{G}\in \gr{\hat{G}}\\G_1,\ldots,G_m\in\gr{G}}} \!\!\!\!\!\chi_{\hat{\phi}}(\hat{G})|{\hat{\phi}}|\bra{Q}\Delta_{G_{\mathrm{inv}}}\Delta_{G_m} \ldots \Delta_{G_1\hat{G}}\ket{\rho}
\end{align}
Hence we can write
\begin{equation}
k_m^\lambda=f_{\lambda'}^m \bra{Q}\mc{L} \mc{P}_\phi\mc{R}\ket{\rho} + \varepsilon_m 
\end{equation}
with 
\begin{equation}
\varepsilon_m = \!\!\!\!\!\avg_{\substack{\hat{G}\in \gr{\hat{G}}\\G_1,\ldots,G_m\in\gr{G}}} \!\!\!\!\!\chi_{\hat{\phi}}(\hat{G})|{\hat{\phi}}|\bra{Q}\Delta_{G_{\mathrm{inv}}}\Delta_{G_m} \ldots \Delta_{G_1\hat{G}}\ket{\rho}.
\end{equation}
We can upper bound $\varepsilon_m$ by
\begin{align}
&\!\!\!\!\!\!\!\avg_{\substack{\hat{G}\in \gr{\hat{G}}\\G_1,\ldots,G_m\in\gr{G}}} \!\!\!\!\!\chi_{\hat{\phi}}(\hat{G})|{\hat{\phi}}|\bra{Q}\Delta_{G_{\mathrm{inv}}}\Delta_{G_m} \ldots \Delta_{G_1\hat{G}}\ket{\rho} \\
&\leq\!\!\!\!\!\! \avg_{\substack{\hat{G}\in \gr{\hat{G}}\\G_1,\ldots,G_m\in\gr{G}}}\!\!\!\!\!\!\!\!\!|\chi_{\hat{\phi}}(\hat{G})||{\hat{\phi}}|\norm{\Delta_{G_{\mathrm{inv}}}}_\diamond \norm{\Delta_{G_{m}}}_\diamond\ldots \norm{\Delta_{G_{1}\hat{G}}}_\diamond\\
&\leq \max_{\hat{G}\in \gr{\hat{G}}} |\chi_{\hat{\phi}}(\hat{G})||{\hat{\phi}}|\max_{G\in \gr{G}}\norm{\Delta_G}_\diamond\left(\md{E}_{G\in\gr{G}}\norm{\Delta_G}_\diamond \right)^m.
\end{align}
Setting
\begin{align}
\delta_1 &= |{\hat{\phi}}|\left(\max_{\hat{G}\in \gr{\hat{G}} }|\chi_{\hat{\phi}}(\hat{G})|\right)\left( \max_{G\in \gr{G}} \norm{\Delta_G}_\diamond\right)\\
\delta_2 &= \avg_{G \in \gr{G} }\norm{\Delta_G }_{\diamond} 
\end{align}
we complete the proof.
\end{proof}
In \cite{wallman2018randomized} it was shown that $\delta_2$ is small for realistic gate-dependent noise. This implies that for large enough $m$ the outcome of a character randomized benchmarking experiment can be described by a single exponential decay (up to a small, exponentially decreasing factor). The rate of decay $f_{\lambda'}$ can be related to the largest eigenvalue of the operator $\md{E}_{G\in\gr{G}} (\mc{\tilde{G}}\otimes \phi_{\lambda'}(G))$. We can interpret this rate of decay following Wallman~\cite{wallman2018randomized} by setting w.l.o.g. $\mc{\tilde{G}} = \mc{L}_G\mc{G}\mc{R}$ where $\mc{R}$ is defined as in \cref{thm:eig_ops} and is invertible (we can always render $\mc{R}$ invertible by an arbitrary small perturbation).
Now consider from $\mc{\tilde{G}} = \mc{L}_G\mc{G}\mc{R}$ and the invertibility of $\mc{R}$:

\begin{align}
\avg_{G\in \gr{G}}\tr(\mc{G}\ct \mc{R} \mc{\tilde{G}}\mc{R}^{-1}) &=\avg_{G\in \gr{G}}\tr(\mc{G}\ct \mc{R} \mc{L}_G\mc{G}\mc{R}\mc{R}^{-1})\\& = \avg_{G\in \gr{G}}\tr(\mc{R}\mc{L}_G)
\end{align}
and moreover from \cref{eq:R_eig}:
\begin{align}
\avg_{G\in \gr{G}}\tr(\mc{G}\ct \mc{R} \mc{\tilde{G}}\mc{R}^{-1}) = \sum_{\lambda\in R_{\gr{G}}} f_{\lambda} \tr(\mc{P}_\lambda).
\end{align}
From this we can consider the average fidelity of noise \emph{between gates} (the map $\mc{R}\mc{L}_G)$ averaged over all gates:
\begin{align}
\avg_{G\in \gr{G}}F_{\mathrm{avg}}(\mc{R}\mc{L}_G) &= \!\avg_{G\in \gr{G}}\frac{2^{-q}\tr(\mc{R}\mc{L}_G)+1}{2^q+1} \\&= \!\frac{2^{-q}\sum_{\lambda\in R_{\gr{G}}}f_{\lambda}\tr(\mc{P}_\lambda)+1}{2^q+1}.
\end{align}
Hence can interpret the quality parameters given by character randomized benchmarking as characterizing the average noise in between gates, extending the conclusion reached in \cite{wallman2018randomized} for standard randomized benchmarking to character randomized benchmarking. In \cite{merkel2018randomized} an alternative interpretation of the decay rate of randomized benchmarking in the presence of gate dependent noise is given in terms of Fourier transforms of matrix valued group functions. One could recast the above analysis for character randomized benchmarking in this language as well but we do not pursue this further here.\\

\noindent{\bf Interleaved character randomized benchmarking}
\noindent In the main text we proposed 2-for-1 interleaved randomized benchmarking, a form of character interleaved randomized benchmarking. More generally we can consider performing interleaved character randomized benchmarking with a benchmarking group $\gr{G}$, a character group $\gr{\hat{G}}$, and an interleaving gate $C$. However it is not obvious that the interleaved character randomized benchmarking procedure (for arbitrary $\gr{G}$ and $C$) always yields data that can be fitted to a single exponential such that the average fidelity can be extracted. Here we will justify this behavior subject to an assumption on the relation between the interleaving gate $C$ and the benchmarking group $\gr{G}$ which we expect to be quite general. This relation is phrased in terms of what we call the `mixing matrix' of the group $\gr{G}$ and gate $C$. This matrix, which we denote by $M$, has entries 
\begin{equation}
M_{\lambda,\hat{\lambda}} = \frac{1}{\tr(\mc{P}_\lambda)}\tr\left(\mc{P}_\lambda  \mc{C}\mc{P}_{\hat{\lambda}}\mc{C}\ct\right)
\end{equation}
for $\lambda,\lambda'\in R'_{\gr{G}} =R_{\gr{G}}\backslash\{\mathrm{id}\} $ with $\phi_{\mathrm{id}}$ the trivial subrepresentation of the PTM representation of $\gr{G}$ carried by $\ket{\id}$  and where $\mc{P}_\lambda$ is the projector onto the subrepresentation $\phi_\lambda$ of $\mc{G}$. Note that this matrix is defined completely by $C$ and the PTM representation of $\gr{G}$. Note also that this matrix has only non-negative entries, that is $M_{\lambda,\hat{\lambda}}\geq 0 \;\; \forall \lambda,\hat{\lambda}$.\\

\noindent In the following lemma we will assume that the mixing matrix $M$ is not only non-negative but also irreducible in the Perron-Frobenius sense~\cite{maccluer2000many}. Formally this means that there exists an integer $L$ such that $A^L$ has only strictly positive entries. This assumption will allow us to invoke the powerful Perron-Frobenius theorem~\cite{maccluer2000many} to prove in \cref{lem:int_char_bench} that interleaved character randomized benchmarking works as advertised. Below \cref{lem:int_char_bench} we will also explicitly verify the irreducibility condition for 2-for-1 interleaved benchmarking with the CPHASE gate. We note that the assumption of irreducibility of $M$ can be easily relaxed to $M$ being a direct sum of irreducible matrices with the proof of \cref{lem:int_char_bench} basically unchanged. It is an open question if it can be relaxed further to encompass all non-negative mixing matrices.
\begin{theorem}\label{lem:int_char_bench}
Consider the outcome  $k_{\lambda'}^m$ of an interleaved character randomized benchmarking experiment benchmarking group $\gr{G}$, character group $\hat{G}$, subrepresentations $\hat{\phi}\subset \phi_{\lambda'}$ for some $\lambda'\in R_{\gr{G}}$, interleaving gate $C$, and set of sequence lengths $\md{M}$ and assume the existence of quantum channels $\mc{E}_C, \mc{E}$ s.t. $\mc{\widetilde{C}} = \mc{C}\mc{E}_C$ and $\mc{\widetilde{G}} = \mc{E}\mc{G}$ for all $G\in \gr{G}$. Now consider the matrix $M(\mc{E}_C\mc{E})$ as a function of the composed channel $\mc{E}_C\mc{E}$ with entries
\begin{equation}
M_{\lambda,\hat{\lambda}}(\mc{E}_C\mc{E}) = \frac{1}{\tr(\mc{P}_\lambda)}\tr\left(\mc{P}_\lambda  \mc{C}\mc{P}_{\hat{\lambda}}\mc{C}\ct \mc{E}_C\mc{E}\right)
\end{equation}
for $\lambda,\lambda'\in R'_{\gr{G}}=R_{\gr{G}}\backslash\{\mathrm{id}\}$ where $\mc{P}_\lambda$ is again the projector onto the subrepresentation $\phi_\lambda$ of $\mc{G}$. If for $\mc{E}=\mc{E}_C  = \mc{I}$ (the identity map) the matrix $M(\mc{I}) = M$ (the mixing matrix defined above) is irreducible (in the sense of Perron-Frobenius), then there exist parameters $A,f_{\lambda'}$ s.t. 
\begin{equation}
|k_{\lambda'}^m  - Af_{\lambda'}^m| \leq \delta_1\delta_2^m
\end{equation}
with $\delta_1 = O(1-F_{\mathrm{avg}}(\mc{E}_C\mc{E}))$ and $\delta_2 = \gamma + O([1- F_{\mathrm{avg}}(\mc{E}_C\mc{E})]^2)$ where $\gamma$ is the second largest eigenvalue (in absolute value) of $M$.
Moreover we have that (noting that $f_{\mathrm{id}} =1$ as the map $\mc{E}_C\mc{E}$ is CPTP):
\begin{align}
&\left|\frac{1}{2^q}\sum_{\lambda\in R_{\gr{G}}} \tr(\mc{P}_\lambda) f_{\lambda} - \frac{2^q (F_{\mathrm{avg}}(\mc{E}_C\mc{E}) +1)}{2^q+1} \right|\notag\\
&\hspace{10em}\leq O\left([1-F_{\mathrm{avg}}(\mc{E}_C\mc{E})]^2\right)
\end{align}
\end{theorem}
\begin{proof}
Consider the definition of $k_{\lambda'}^m$:
\begin{align}
k^{\lambda'}_m &= |\hat{\phi}|\!\!\!\!\!\!\avg_{\substack{\hat{G}\in \gr{\hat{G}}\\G_1,\ldots,G_m\in\gr{G}}}\!\!\!\!\!\!\chi_{\hat{\phi}}(\hat{G})\bra{Q}\mc{E}_{\mathrm{inv}}\mc{G}_{\mathrm{inv}}\mc{C}\mc{E}_C\mc{E}\mc{G}_m\notag\\
&\hspace{8em}\times\mc{C}\mc{E}_C\mc{E}\ldots \mc{C}\mc{E}_C\mc{E}\mc{G}_1\mc{\hat{G}}\ket{\rho},
\end{align}
where $G_{\mathrm{inv}} = G_1\ct C\ct \cdots G_m\ct C\ct $ and $\mc{E}_{\mathrm{inv}}$ is the noise associated to the inverse gate (which we assume to be constant).
Using the character projection formula and Schur's lemma we can write this as
\begin{align}
k^{\lambda'}_m &=\avg_{G_1,\ldots,G_{m-1}\in \gr{G}}\bra{Q}\mc{E}_{\mathrm{inv}}\mc{G}_1\ct\mc{C}\ct\cdots \mc{G}_{m-1}\ct\mc{C}\ct\notag\\
&\hspace{2em}\times\left[\sum_{\lambda_m\in R'_{\gr{G}}} \frac{\tr(P_{\lambda_m}\mc{E}_C\mc{E})}{\tr(\mc{P}_{\lambda_m})}\mc{P}_{\lambda_m}\right]\mc{C}\mc{E}_C\mc{E}\mc{G}_{m-1}\notag\\
&\hspace{6em}\times \mc{C}\mc{E}_C\mc{E}\ldots \mc{C}\mc{E}_C\mc{E}\mc{G}_1\mc{P}_{{\hat{\phi}}}\ket{\rho}.
\end{align}
Note now that in general $\mc{C}$ and $\mc{P}_{\lambda_m}$ do not commute. This means that we can not repeat the reasoning of \cref{lem:rand_bench_av} but must instead write (using Schur's lemma again):
\begin{align}
k^{\lambda'}_m &=\sum_{\lambda_m\in R'_{\gr{G}}}\frac{\tr(P_{\lambda_m}\mc{E}_C\mc{E})}{\tr(\mc{P}_{\lambda_m})} \!\!\!\!\!\avg_{G_1,\ldots,G_{m-2}\in \gr{G}}\!\!\!\!\!\!\bra{Q}\mc{E}_{\mathrm{inv}}\mc{G}_1\ct\mc{C}\ct\cdots \mc{G}_{m-2}\ct\mc{C}\ct\notag\\&\hspace{20mm}\times \left[\sum_{\lambda_{m-1}\in R'_{\gr{G}}} \frac{\tr(\mc{P}_{\lambda_{m-1}}\mc{C}\ct\mc{P}_{\lambda_m}\mc{C}\mc{E}_C\mc{E})}{\tr(\mc{P}_{\lambda_{m-1}})}\right]\notag\\
&\hspace{30mm}\times \mc{P}_{\lambda_{m-1}}     \mc{C}\mc{E}_C\mc{E}\mc{G}_{m-2} \mc{C}\mc{E}_C\notag\\
&\hspace{35mm}\times\mc{E}\ldots \mc{C}\mc{E}_C\mc{E}\mc{G}_1\mc{P}_{{\hat{\phi}}}\ket{\rho}.
\end{align}
Here we recognize the definition of the matrix element $M_{\lambda_{m-1},\lambda_m}(\mc{E}_C\mc{E})$. Moreover we can apply the above expansion to $G_{m-2}, G_{m-3}$ and so forth writing the result in terms of powers of the matrix $M(\mc{E}_C\mc{E})$. After some reordering we get
\begin{equation*}
k^{\lambda'}_m =\sum_{\lambda_1, \lambda_{m}\in R'_{\gr{G}}}\frac{\tr(P_{\lambda_m}\mc{E}_C\mc{E})}{\tr(\mc{P}_{\lambda_m})} [M^{m-1}]_{\lambda_1,\lambda_m}  \bra{Q}\mc{P}_{\lambda_1}\mc{P}_{\hat{\phi}}\ket{\rho}
\end{equation*}
where we have again absorbed the noise associated with the inverse $G_{\mathrm{inv}}$ into the measurement POVM element $Q$. Now recognizing that by construction $\mc{P}_{\hat{\phi}}\subset \mc{P}_{\lambda'}$ we can write $k^{\lambda'}_m$ as 
\begin{equation}\label{eq:matrix_form}
k^{\lambda'}_m = e_{\lambda'}M^m v^{T} \bra{Q}\mc{P}_{\hat{\phi}}\ket{\rho}
\end{equation}
where $e_{\lambda'}$ is the $\lambda'th$ standard basis row vector of length $R'_{\gr{G}}$ and $v = v(\mc{E}_C\mc{E})$ is a row vector of length $R'_{\gr{G}}$ with entries $[v]_{\lambda} =\frac{\tr(P_{\lambda_m}\mc{E}_C\mc{E})}{\tr(\mc{P}_{\lambda_m})}$. This looks somewhat like an exponential decay but not quite. Ideally we would like that $M^m$ has one dominant eigenvalue and moreover that the vector $v$ has high overlap with the corresponding eigenvector. This would guarantee that $k^{\lambda'}_m$ is close to a single exponential. The rest of the proof will argue that this is indeed the case. Now we use the assumption of the irreducibility of the mixing matrix $M = M(\mc{I})$. Subject to this assumption, the Perron-Frobenius theorem~\cite{maccluer2000many} states that the matrix $M$ has a non-degenerate eigenvalue $\gamma_{\mathrm{max}}(M(\mc{I}))$ that is strictly larger in absolute value than all other eigenvalues of $M(\mc{I})$ and moreover satisfies the inequality
\begin{equation}
\min_{\lambda\in R'_{\gr{G}}}\sum_{\hat{\lambda}\in R'_{\gr{G}}} M_{\lambda,\hat{\lambda}} \leq \gamma_{\mathrm{max}}(M(\mc{I})) \leq \max_{\lambda\in R'_{\gr{G}}}\sum_{\hat{\lambda}\in R'_{\gr{G}}} M_{\lambda,\hat{\lambda}}.
\end{equation}
It is easy to see from the definition of $M_{\lambda,\hat{\lambda}}$ that
\begin{align}
\sum_{\hat{\lambda}\in R'_{\gr{G}}} M_{\lambda,\hat{\lambda}} &= \sum_{\hat{\lambda}\in R'_{\gr{G}}}\frac{1}{\tr(\mc{P}_\lambda)}\tr\left(\mc{P}_\lambda  \mc{C}\mc{P}_{\hat{\lambda}}\mc{C}\ct\right)\\
&=\sum_{\hat{\lambda}\in R'_{\gr{G}}}\frac{1}{\tr(\mc{P}_\lambda)} \left(\mc{P}_\lambda  \mc{C}\sum_{\hat{\lambda}\in R'_{\gr{G}}}\mc{P}_{\hat{\lambda}}\mc{C}\ct\right)\\
& =\frac{\tr(\mc{P}_\lambda)}{\tr(\mc{P}_\lambda)}=1
\end{align}
for all $\lambda\in R'_{\gr{G}}$.  This means the largest eigenvalue of $M(\mc{I})$ is exactly $1$. Moreover, as one can easily deduce by direct calculation, the associated right-eigenvector is the vector $v^R = (1,1,\ldots,1)$. Note that this vector is precisely $v(\mc{E}_C\mc{E})$ (as defined in \cref{eq:matrix_form}) for $\mc{E}_C\mc{E} = \mc{I}$. Similarly the left-eigenvector of $M = M(\mc{I})$ is given by (in terms of its components) $v^L _\lambda = \tr(\mc{P}_\lambda)$. This allows us to calculate that $k^{\lambda'}_m =\bra{Q}\mc{P}_{\hat{\phi}}\ket{\rho}$ if $\mc{E}_C\mc{E} = \mc{I}$, which is as expected.\\

\noindent Now we will consider the map $\mc{E}_C\mc{E}$ as a perturbation of $\mc{I}$ with the perturbation parameter
\begin{equation}
\alpha = 1 - \frac{\tr(\mc{P}_{\mathrm{tot}}\mc{E}_C\mc{E})}{\tr(\mc{P}_{\mathrm{tot}})}
\end{equation}
with $\mc{P}_{\mathrm{tot}}= \sum_{\lambda\in R'_{\gr{G}}}\mc{P}_\lambda$. We can write the quantum channel $\mc{E}_C\mc{E}$ as $\mc{E}_C\mc{E} = \mc{I} -\alpha \mc{F}$ where $\mc{F}$ is some superoperator (not CP, but by construction trace-annihilating). Since $M(\mc{E}_C\mc{E})$ is linear in its argument we can write $M(\mc{E}_C\mc{E}) = M(\mc{I}) - \alpha M(\mc{F})$. From standard matrix perturbation theory~\cite[Section 5.1]{sakurai2014modern} we can approximately calculate the largest eigenvalue of $M(\mc{E}_C\mc{E})$ as
\begin{align}\label{eq:pert_eig}
\gamma_{\mathrm{max}}(M(\mc{E}_C\mc{E})) = \gamma_{\mathrm{max}}&(M(\mc{I})\notag\\& - \alpha \frac{v^L M(\mc{F}){v^R}^T}{v^L{v^R}^T} + O(\alpha^2)
\end{align}
We can now calculate the prefactor $\frac{v^L M(\mc{F}){v^R}^T}{v^L{v^R}^T}$ as
\begin{align}
\frac{v^L A(\mc{F}){v^R}^T}{v^L{v^R}^T} &= \frac{\sum_{\lambda\in R'_{\gr{G}}}\sum_{\hat{\lambda}\in R'_{\gr{G}}}v^L_\lambda M(\mc{F})_{\lambda,\hat{\lambda}}v^R_{\hat{\lambda}}}{\tr(\mc{P}_{\mathrm{tot}})}\\
& = \frac{\sum_{\lambda\in R'_{\gr{G}}}\sum_{\hat{\lambda}\in R'_{\gr{G}}}\tr(P_{\lambda}C\ct \mc{P}_{\hat{\lambda}}\mc{F})}{\tr(\mc{P}_{\mathrm{tot}})}\\
&= \frac{\tr\left(\mc{P}_{\mathrm{tot}}\mc{F}\right)}{\tr(\mc{P}_{\mathrm{tot}})}\\
&= \frac{1}{\alpha} \frac{\tr\left(\mc{P}_{\mathrm{tot}}[\mc{I}-\mc{E}_C\mc{E}]\right)}{\tr(\mc{P}_{\mathrm{tot}})}\\
& =1
\end{align}
where we used the definition of $\alpha$ in the last line. This means that $\gamma_{\mathrm{max}}(M(\mc{E}_C\mc{E})) = 1-\alpha$ up to $O(\alpha^2)$corrections. One could in principle calculate the prefactor of the correction term, but we will not pursue this here.  Now we know that the matrix $M(\mc{E}_C\mc{E})^{m-1}$ in \cref{eq:matrix_form} will be dominated by a factor $(1-\alpha + O(\alpha^2))^{m-1}$. However it could still be that the vector $v(\mc{E}_C\mc{E})$ in \cref{eq:matrix_form} has small overlap with the right-eigenvector $v^R(\mc{E}_C\mc{E})$ of $M(\mc{E}_C\mc{E})$ associated to the largest eigenvalue $\gamma_{\mathrm{max}}(M(\mc{E}_C\mc{E}))$. We can again use a perturbation argument to see that this overlap will be big. Again from standard perturbation theory~\cite[Section 5.1]{sakurai2014modern} we have 
\begin{align}
\norm{v^R(\mc{E}_C\mc{E}) -v^R(\mc{I}) } = O(|\alpha|).
\end{align}
Moreover, by definition of $v^R(\mc{I})$ and $v(\mc{E}_C\mc{E})$ we have that $v^Rv(\mc{E}_C\mc{E})^T = 1- \alpha$. By the triangle inequality we thus have
\begin{equation}
\norm{v^R(\mc{E}_C\mc{E}) - v(\mc{E}_C\mc{E})} = O(|\alpha|).
\end{equation}
One can again fill in the constant factors here if one desires a more precise statement.
Finally we note from \cref{lem:trace_to_fid} that
\begin{equation}\label{eq:al_fid}
\alpha = 1- \frac{\tr(\mc{P}_{\mathrm{tot}}\mc{E}_C\mc{E})}{\tr(\mc{P}_{\mathrm{tot}})} = \frac{2^q }{2^q-1}(F(\mc{E}_C\mc{E})-1)
\end{equation}
This means that in the relevant limit of high fidelity, $\alpha$ will be small, justifying our perturbative analysis.
Defining $\gamma$ to be the second largest (in absolute value) eigenvalue of $M(\mc{E}_C\mc{E})$, which by the same argument as above will be the second largest eigenvalue of $M(\mc{I})$ up to $O(\alpha^2)$ corrections, we get
\begin{equation*}
|k_m^{\lambda'} - \bra{Q}\mc{P}_{\hat{\phi}}\ket{\rho} -\gamma_{\mathrm{max}}(M(\mc{E}_C\mc{E}))^{m-1}   \bra{Q}\mc{P}_{\hat{\phi}}\ket{\rho}|  \leq \delta_1\delta_2^m
\end{equation*}
with $\delta_1 = O(1-F_{\mathrm{avg}}(\mc{E}_C\mc{E}))$ and $\delta_2 = |\gamma| + O((1- F_{\mathrm{avg}}(\mc{E}_C\mc{E}))^2)$. Moreover, we have from \cref{eq:pert_eig,eq:al_fid} that 
\begin{align}
\gamma_{\mathrm{max}}(A(\mc{E}_C\mc{E}))  &= 1-\frac{2^q }{2^q-1}(F(\mc{E}_C\mc{E})-1) \notag\\
&\hspace{5em}+ O\left([1-F_{\mathrm{avg}}(\mc{E}_C\mc{E})]^2\right)
\end{align}
which immediately implies
\begin{align}
&\left|\frac{1}{2^q}\sum_{\lambda\in R_{\gr{G}}} \tr(\mc{P}_\lambda) f_{\lambda} - \frac{2^q (F_{\mathrm{avg}}(\mc{E}_C\mc{E}) +1)}{2^q+1} \right|\notag\\
&\hspace{10em}\leq O\left([1-F_{\mathrm{avg}}(\mc{E}_C\mc{E})]^2\right)
\end{align}
proving the lemma.
\end{proof}

\noindent It is instructive to calculate the mixing matrix for a relevant example. We will calculate $M$ for $C$ the CPHASE gate and $\gr{G} = \gr{C}_1\tn{2}$ two copies of the single qubit Clifford gates. Recall from the main text that the PTM representation of $\gr{C}_1\tn{2}$ has three non-trivial subrepresentations. From their definitions in \cref{eq:2for1_reps} and the action of the CPHASE gate on the two qubit Pauli operators it is straightforward to see that the mixing matrix is of the form
\begin{equation}
M = \begin{pmatrix} 1/3 & 0 & 2/3 \\ 0 & 1/3 & 2/3 \\ 2/9 & 2/9 & 5/9 \end{pmatrix}.
\end{equation}
Calculating $M^2$ one can see that $M$ is indeed irreducible. Moreover $M$ has eigenvalues $1, 1/3$ and $-1/9$. This means that for 2-for-1 interleaved benchmarking the interleaved experiment produces data that deviates from a single exponential no more than $(1/3)^m$ (for sufficiently high fidelity) which will be negligible for even for fairly small $m$. This means that for 2-for-1 interleaved benchmarking the assumption that the interleaved experiment produces data described by a single exponential is good. We will see this confirmed numerically in the simulated experiment presented in Supplementary fig. 2. Finally we note that a similar result was achieved using different methods in \cite{erhard2019characterizing}.

\section{ Data Availability Statement} The data and analysis used to generate Supplementary fig. 2 will be available online at\\ https://doi.org/10.5281/zenodo.2549368 . No other supporting data was generated or analysed for this work.

\section{Competing Interests} The authors declare that there are no competing interests.
\section{Author Contributions} JH, XX, LMKV and SW conceived of the theoretical framework, detailed analysis was done by JH with input from XX, LMKV and SW, JH wrote the manuscript with input from XX, LMKV and SW, SW supervised the project.
\section{Acknowledgements} The authors would like to thank Thomas F. Watson, J\'er\'emy Ribeiro and Bas Dirkse for enlightening discussions. While preparing a new version of this manuscript the authors became aware of similar, independent work by Wallman \& Emerson~\cite{wallman2018determining}. JH and SW are funded by STW Netherlands, NWO VIDI, an ERC Starting Grant and by the NWO Zwaartekracht QSC grant. XX and LMKV are funded by the Army Research Office (ARO) under Grant Number W911NF-17-1-0274.

 % \end{document} %usually commented out, here for wordcounting purposes
 %3565 words + 11 equations (16 words each) making 3741 words
\bibliographystyle{naturemag} % Tell bibtex which bibliography style to use
 \bibliography{characterRBlibrary}

\begin{thebibliography}{10}
\expandafter\ifx\csname url\endcsname\relax
  \def\url#1{\texttt{#1}}\fi
\expandafter\ifx\csname urlprefix\endcsname\relax\def\urlprefix{URL }\fi
\providecommand{\bibinfo}[2]{#2}
\providecommand{\eprint}[2][]{\url{#2}}

\bibitem{dankert2006c}
\bibinfo{author}{Dankert, C.}
\newblock \bibinfo{title}{C. r, e. j, livine e. exact and approximate unitary
  2-designs: Constructions and applications}.
\newblock \emph{\bibinfo{journal}{Phys. Rev. A}} \textbf{\bibinfo{volume}{80}},
  \bibinfo{pages}{012304} (\bibinfo{year}{2006}).

\bibitem{Magesan_2012}
\bibinfo{author}{Magesan, E.}, \bibinfo{author}{Gambetta, J.~M.} \&
  \bibinfo{author}{Emerson, J.}
\newblock \bibinfo{title}{Characterizing quantum gates via randomized
  benchmarking}.
\newblock \emph{\bibinfo{journal}{Phys. Rev. A}} \textbf{\bibinfo{volume}{85}}
  (\bibinfo{year}{2012}).

\bibitem{Emerson2005}
\bibinfo{author}{Emerson, J.}, \bibinfo{author}{Alicki, R.} \&
  \bibinfo{author}{\.{Z}yczkowski, K.}
\newblock \bibinfo{title}{{Scalable noise estimation with random unitary
  operators}}.
\newblock \emph{\bibinfo{journal}{J. Opt. B}} \textbf{\bibinfo{volume}{7}},
  \bibinfo{pages}{S347} (\bibinfo{year}{2005}).

\bibitem{Chow2009}
\bibinfo{author}{Chow, J.~M.} \emph{et~al.}
\newblock \bibinfo{title}{Randomized benchmarking and process tomography for
  gate errors in a solid-state qubit}.
\newblock \emph{\bibinfo{journal}{Phys. Rev. Lett.}}
  \textbf{\bibinfo{volume}{102}}, \bibinfo{pages}{090502}
  (\bibinfo{year}{2009}).

\bibitem{Gaebler2012}
\bibinfo{author}{Gaebler, J.~P.} \emph{et~al.}
\newblock \bibinfo{title}{{Randomized Benchmarking of Multiqubit Gates}}.
\newblock \emph{\bibinfo{journal}{Phys. Rev. Lett.}}
  \textbf{\bibinfo{volume}{108}}, \bibinfo{pages}{260503}
  (\bibinfo{year}{2012}).

\bibitem{Granade2014}
\bibinfo{author}{Granade, C.}, \bibinfo{author}{Ferrie, C.} \&
  \bibinfo{author}{Cory, D.~G.}
\newblock \bibinfo{title}{{Accelerated Randomized Benchmarking}}.
\newblock \emph{\bibinfo{journal}{New J. Phys.}} \textbf{\bibinfo{volume}{17}},
  \bibinfo{pages}{013042} (\bibinfo{year}{2014}).
\newblock \eprint{1404.5275v1}.

\bibitem{Epstein2014}
\bibinfo{author}{{Epstein}, J.~M.}, \bibinfo{author}{{Cross}, A.~W.},
  \bibinfo{author}{{Magesan}, E.} \& \bibinfo{author}{{Gambetta}, J.~M.}
\newblock \bibinfo{title}{{Investigating the limits of randomized benchmarking
  protocols}}.
\newblock \emph{\bibinfo{journal}{\pra}} \textbf{\bibinfo{volume}{89}},
  \bibinfo{pages}{062321} (\bibinfo{year}{2014}).
\newblock \eprint{1308.2928}.

\bibitem{Knill2008}
\bibinfo{author}{Knill, E.} \emph{et~al.}
\newblock \bibinfo{title}{{Randomized benchmarking of quantum gates}}.
\newblock \emph{\bibinfo{journal}{Phys. Rev. A}} \textbf{\bibinfo{volume}{77}},
  \bibinfo{pages}{012307} (\bibinfo{year}{2008}).

\bibitem{Asaad2016}
\bibinfo{author}{Asaad, S.} \emph{et~al.}
\newblock \bibinfo{title}{{Independent, extensible control of same-frequency
  superconducting qubits by selective broadcasting}}.
\newblock \emph{\bibinfo{journal}{npj Quantum Inf.}}
  \textbf{\bibinfo{volume}{2}}, \bibinfo{pages}{16029} (\bibinfo{year}{2016}).
\newblock \eprint{1508.06676}.

\bibitem{Barends2014}
\bibinfo{author}{Barends, R.} \emph{et~al.}
\newblock \bibinfo{title}{{Superconducting quantum circuits at the surface code
  threshold for fault tolerance.}}
\newblock \emph{\bibinfo{journal}{Nature}} \textbf{\bibinfo{volume}{508}},
  \bibinfo{pages}{500--3} (\bibinfo{year}{2014}).

\bibitem{DiCarlo2009}
\bibinfo{author}{DiCarlo, L.} \emph{et~al.}
\newblock \bibinfo{title}{{Demonstration of two-qubit algorithms with a
  superconducting quantum processor.}}
\newblock \emph{\bibinfo{journal}{Nature}} \textbf{\bibinfo{volume}{460}},
  \bibinfo{pages}{240} (\bibinfo{year}{2009}).
\newblock \eprint{0903.2030}.

\bibitem{o2015qubit}
\bibinfo{author}{O’Malley, P.} \emph{et~al.}
\newblock \bibinfo{title}{Qubit metrology of ultralow phase noise using
  randomized benchmarking}.
\newblock \emph{\bibinfo{journal}{Phys. Rev. Applied}}
  \textbf{\bibinfo{volume}{3}}, \bibinfo{pages}{044009} (\bibinfo{year}{2015}).

\bibitem{sheldon2016characterizing}
\bibinfo{author}{Sheldon, S.} \emph{et~al.}
\newblock \bibinfo{title}{Characterizing errors on qubit operations via
  iterative randomized benchmarking}.
\newblock \emph{\bibinfo{journal}{Phys. Rev. A}} \textbf{\bibinfo{volume}{93}},
  \bibinfo{pages}{012301} (\bibinfo{year}{2016}).

\bibitem{wallman2018randomized}
\bibinfo{author}{Wallman, J.~J.}
\newblock \bibinfo{title}{Randomized benchmarking with gate-dependent noise}.
\newblock \emph{\bibinfo{journal}{Quantum}} \textbf{\bibinfo{volume}{2}},
  \bibinfo{pages}{47} (\bibinfo{year}{2018}).

\bibitem{proctor2017randomized}
\bibinfo{author}{Proctor, T.}, \bibinfo{author}{Rudinger, K.},
  \bibinfo{author}{Young, K.}, \bibinfo{author}{Sarovar, M.} \&
  \bibinfo{author}{Blume-Kohout, R.}
\newblock \bibinfo{title}{What randomized benchmarking actually measures}.
\newblock \emph{\bibinfo{journal}{Phys. Rev. Lett.}}
  \textbf{\bibinfo{volume}{119}}, \bibinfo{pages}{130502}
  (\bibinfo{year}{2017}).

\bibitem{merkel2018randomized}
\bibinfo{author}{Merkel, S.~T.}, \bibinfo{author}{Pritchett, E.~J.} \&
  \bibinfo{author}{Fong, B.~H.}
\newblock \bibinfo{title}{Randomized benchmarking as convolution: Fourier
  analysis of gate dependent errors}.
\newblock \emph{\bibinfo{journal}{arXiv preprint arXiv:1804.05951}}
  (\bibinfo{year}{2018}).

\bibitem{Nielsen2011}
\bibinfo{author}{Nielsen, M.~A.} \& \bibinfo{author}{Chuang, I.~L.}
\newblock \emph{\bibinfo{title}{Quantum Computation and Quantum Information:
  10th Anniversary Edition}} (\bibinfo{publisher}{Cambridge University Press},
  \bibinfo{address}{New York, NY, USA}, \bibinfo{year}{2011}),
  \bibinfo{edition}{10th} edn.

\bibitem{Magesan2012a}
\bibinfo{author}{Magesan, E.}, \bibinfo{author}{Gambetta, J.~M.} \&
  \bibinfo{author}{Emerson, J.}
\newblock \bibinfo{title}{{Characterizing quantum gates via randomized
  benchmarking}}.
\newblock \emph{\bibinfo{journal}{Phys. Rev. A}} \textbf{\bibinfo{volume}{85}},
  \bibinfo{pages}{042311} (\bibinfo{year}{2012}).

\bibitem{Cross_2016}
\bibinfo{author}{Cross, A.~W.}, \bibinfo{author}{Magesan, E.},
  \bibinfo{author}{Bishop, L.~S.}, \bibinfo{author}{Smolin, J.~A.} \&
  \bibinfo{author}{Gambetta, J.~M.}
\newblock \bibinfo{title}{Scalable randomised benchmarking of non-clifford
  gates}.
\newblock \emph{\bibinfo{journal}{npj Quantum Information}}
  \textbf{\bibinfo{volume}{2}} (\bibinfo{year}{2016}).

\bibitem{brown2018randomized}
\bibinfo{author}{Brown, W.~G.} \& \bibinfo{author}{Eastin, B.}
\newblock \bibinfo{title}{Randomized benchmarking with restricted gate sets}.
\newblock \emph{\bibinfo{journal}{arXiv preprint arXiv:1801.04042}}
  (\bibinfo{year}{2018}).

\bibitem{hashagen2018real}
\bibinfo{author}{Hashagen, A.}, \bibinfo{author}{Flammia, S.},
  \bibinfo{author}{Gross, D.} \& \bibinfo{author}{Wallman, J.}
\newblock \bibinfo{title}{Real randomized benchmarking}.
\newblock \emph{\bibinfo{journal}{arXiv preprint arXiv:1801.06121}}
  (\bibinfo{year}{2018}).

\bibitem{francca2018approximate}
\bibinfo{author}{Fran{\c{c}}a, D.~S.} \& \bibinfo{author}{Hashagen, A.-L.}
\newblock \bibinfo{title}{Approximate randomized benchmarking for finite
  groups}.
\newblock \emph{\bibinfo{journal}{arXiv preprint arXiv:1803.03621}}
  (\bibinfo{year}{2018}).

\bibitem{Dugas2015}
\bibinfo{author}{Dugas, A.~C.}, \bibinfo{author}{Wallman, J.~J.} \&
  \bibinfo{author}{Emerson, J.}
\newblock \bibinfo{title}{{Characterizing Universal Gate Sets via Dihedral
  Benchmarking}}.
\newblock \emph{\bibinfo{journal}{arXiv preprint arXiv:1508.06312}}
  \eprint{1508.06312}.

\bibitem{Gambetta_2012_sim}
\bibinfo{author}{Gambetta, J.~M.} \emph{et~al.}
\newblock \bibinfo{title}{Characterization of addressability by simultaneous
  randomized benchmarking}.
\newblock \emph{\bibinfo{journal}{Phys. Rev. Lett.}}
  \textbf{\bibinfo{volume}{109}} (\bibinfo{year}{2012}).

\bibitem{harper2017estimating}
\bibinfo{author}{Harper, R.} \& \bibinfo{author}{Flammia, S.~T.}
\newblock \bibinfo{title}{Estimating the fidelity of t gates using standard
  interleaved randomized benchmarking}.
\newblock \emph{\bibinfo{journal}{Quantum Science and Technology}}
  \textbf{\bibinfo{volume}{2}}, \bibinfo{pages}{015008} (\bibinfo{year}{2017}).

\bibitem{Flammia2011}
\bibinfo{author}{Flammia, S.~T.} \& \bibinfo{author}{Liu, Y.-K.}
\newblock \bibinfo{title}{{Direct Fidelity Estimation from Few Pauli
  Measurements}}.
\newblock \emph{\bibinfo{journal}{Phys. Rev. Lett.}}
  \textbf{\bibinfo{volume}{106}}, \bibinfo{pages}{230501}
  (\bibinfo{year}{2011}).

\bibitem{harper2019fault}
\bibinfo{author}{Harper, R.} \& \bibinfo{author}{Flammia, S.~T.}
\newblock \bibinfo{title}{Fault-tolerant logical gates in the ibm quantum
  experience}.
\newblock \emph{\bibinfo{journal}{Physical Review Letters}}
  \textbf{\bibinfo{volume}{122}}, \bibinfo{pages}{080504}
  (\bibinfo{year}{2019}).

\bibitem{Muhonen2015}
\bibinfo{author}{{Muhonen}, J.~T.} \emph{et~al.}
\newblock \bibinfo{title}{{Quantifying the quantum gate fidelity of single-atom
  spin qubits in silicon by randomized benchmarking}}.
\newblock \emph{\bibinfo{journal}{Journal of Physics Condensed Matter}}
  \textbf{\bibinfo{volume}{27}}, \bibinfo{pages}{154205}
  (\bibinfo{year}{2015}).

\bibitem{helsen2017multi}
\bibinfo{author}{Helsen, J.}, \bibinfo{author}{Wallman, J.~J.},
  \bibinfo{author}{Flammia, S.~T.} \& \bibinfo{author}{Wehner, S.}
\newblock \bibinfo{title}{Multi-qubit randomized benchmarking using few
  samples}.
\newblock \emph{\bibinfo{journal}{arXiv preprint arXiv:1701.04299}}
  (\bibinfo{year}{2017}).

\bibitem{Fulton2004}
\bibinfo{author}{Fulton, W.} \& \bibinfo{author}{Harris, J.}
\newblock \emph{\bibinfo{title}{Representation Theory: A First Course}}.
\newblock Readings in Mathematics (\bibinfo{publisher}{Springer-Verlag New
  York}, \bibinfo{year}{2004}).

\bibitem{Note1}
\bibinfo{note}{This representation is also sometimes called the Liouville
  representation or affine representation of quantum channels.~\cite
  {Wallman2014,Wolf2012}}.

\bibitem{Chuang1997}
\bibinfo{author}{Chuang, I.~L.} \& \bibinfo{author}{Nielsen, M.~A.}
\newblock \bibinfo{title}{{Prescription for experimental determination of the
  dynamics of a quantum black box}}.
\newblock \emph{\bibinfo{journal}{J. Mod. Opt.}} \textbf{\bibinfo{volume}{44}},
  \bibinfo{pages}{2455} (\bibinfo{year}{1997}).

\bibitem{Note2}
\bibinfo{note}{Generally the character function is a map to the complex
  numbers, but in our case it is enough to only consider real representations.}

\bibitem{Note3}
\bibinfo{note}{It is straightforward to extend character randomized
  benchmarking to also cover the presence of equivalent irreducible
  subrepresentation. However do not make this extension explicit here in the
  interest of simplicity}.

\bibitem{xue2018benchmarking}
\bibinfo{author}{Xue, X.} \emph{et~al.}
\newblock \bibinfo{title}{Benchmarking gate fidelities in a si/sige two-qubit
  device}.
\newblock \emph{\bibinfo{journal}{arXiv preprint arXiv:1811.04002}}
  (\bibinfo{year}{2018}).

\bibitem{Magesan_2012_interleaved}
\bibinfo{author}{Magesan, E.} \emph{et~al.}
\newblock \bibinfo{title}{Efficient measurement of quantum gate error by
  interleaved randomized benchmarking}.
\newblock \emph{\bibinfo{journal}{Phys. Rev. Lett.}}
  \textbf{\bibinfo{volume}{109}} (\bibinfo{year}{2012}).

\bibitem{dugas2016efficiently}
\bibinfo{author}{Dugas, A.~C.}, \bibinfo{author}{Wallman, J.~J.} \&
  \bibinfo{author}{Emerson, J.}
\newblock \bibinfo{title}{Efficiently characterizing the total error in quantum
  circuits}.
\newblock \emph{\bibinfo{journal}{arXiv preprint arXiv:1610.05296}}
  (\bibinfo{year}{2016}).

\bibitem{watson2018programmable}
\bibinfo{author}{Watson, T.} \emph{et~al.}
\newblock \bibinfo{title}{A programmable two-qubit quantum processor in
  silicon}.
\newblock \emph{\bibinfo{journal}{Nature}} \textbf{\bibinfo{volume}{555}},
  \bibinfo{pages}{633} (\bibinfo{year}{2018}).

\bibitem{zajac2018resonantly}
\bibinfo{author}{Zajac, D.~M.} \emph{et~al.}
\newblock \bibinfo{title}{Resonantly driven cnot gate for electron spins}.
\newblock \emph{\bibinfo{journal}{Science}} \textbf{\bibinfo{volume}{359}},
  \bibinfo{pages}{439--442} (\bibinfo{year}{2018}).

\bibitem{huang2018fidelity}
\bibinfo{author}{Huang, W.} \emph{et~al.}
\newblock \bibinfo{title}{Fidelity benchmarks for two-qubit gates in silicon}.
\newblock \emph{\bibinfo{journal}{arXiv preprint arXiv:1805.05027}}
  (\bibinfo{year}{2018}).

\bibitem{kalra2014robust}
\bibinfo{author}{Kalra, R.}, \bibinfo{author}{Laucht, A.},
  \bibinfo{author}{Hill, C.~D.} \& \bibinfo{author}{Morello, A.}
\newblock \bibinfo{title}{Robust two-qubit gates for donors in silicon
  controlled by hyperfine interactions}.
\newblock \emph{\bibinfo{journal}{Physical Review X}}
  \textbf{\bibinfo{volume}{4}}, \bibinfo{pages}{021044} (\bibinfo{year}{2014}).

\bibitem{li2018crossbar}
\bibinfo{author}{Li, R.} \emph{et~al.}
\newblock \bibinfo{title}{A crossbar network for silicon quantum dot qubits}.
\newblock \emph{\bibinfo{journal}{Science advances}}
  \textbf{\bibinfo{volume}{4}}, \bibinfo{pages}{eaar3960}
  (\bibinfo{year}{2018}).

\bibitem{Wallman2014}
\bibinfo{author}{Wallman, J.~J.} \& \bibinfo{author}{Flammia, S.~T.}
\newblock \bibinfo{title}{{Randomized benchmarking with confidence}}.
\newblock \emph{\bibinfo{journal}{New J. Phys.}} \textbf{\bibinfo{volume}{16}},
  \bibinfo{pages}{103032} (\bibinfo{year}{2014}).

\bibitem{hincks2018bayesian}
\bibinfo{author}{Hincks, I.}, \bibinfo{author}{Wallman, J.~J.},
  \bibinfo{author}{Ferrie, C.}, \bibinfo{author}{Granade, C.} \&
  \bibinfo{author}{Cory, D.~G.}
\newblock \bibinfo{title}{Bayesian inference for randomized benchmarking
  protocols}.
\newblock \emph{\bibinfo{journal}{arXiv preprint arXiv:1802.00401}}
  (\bibinfo{year}{2018}).

\bibitem{Hoeffding1963}
\bibinfo{author}{Hoeffding, W.}
\newblock \bibinfo{title}{Probability inequalities for sums of bounded random
  variables}.
\newblock \emph{\bibinfo{journal}{Journ. Am. Stat. Assoc.}}
  \textbf{\bibinfo{volume}{58}}, \bibinfo{pages}{13--30}
  (\bibinfo{year}{1963}).

\bibitem{watrous2004notes}
\bibinfo{author}{Watrous, J.}
\newblock \bibinfo{title}{Notes on super-operator norms induced by schatten
  norms}.
\newblock \emph{\bibinfo{journal}{arXiv preprint arXiv:0411077}}
  (\bibinfo{year}{2004}).

\bibitem{maccluer2000many}
\bibinfo{author}{MacCluer, C.~R.}
\newblock \bibinfo{title}{The many proofs and applications of perron's
  theorem}.
\newblock \emph{\bibinfo{journal}{Siam Review}} \textbf{\bibinfo{volume}{42}},
  \bibinfo{pages}{487--498} (\bibinfo{year}{2000}).

\bibitem{sakurai2014modern}
\bibinfo{author}{Sakurai, J.~J.}, \bibinfo{author}{Napolitano, J.}
  \emph{et~al.}
\newblock \emph{\bibinfo{title}{Modern quantum mechanics}}, vol.
  \bibinfo{volume}{261} (\bibinfo{publisher}{Pearson}, \bibinfo{year}{2014}).

\bibitem{erhard2019characterizing}
\bibinfo{author}{Erhard, A.} \emph{et~al.}
\newblock \bibinfo{title}{Characterizing large-scale quantum computers via
  cycle benchmarking}.
\newblock \emph{\bibinfo{journal}{arXiv preprint arXiv:1902.08543}}
  (\bibinfo{year}{2019}).

\bibitem{wallman2018determining}
\bibinfo{author}{Wallman, J.~J.} \& \bibinfo{author}{Emerson, J.}
\newblock \bibinfo{title}{Determining the capacity of any quantum computer to
  perform a quantum computation}.
\newblock \emph{\bibinfo{journal}{In Preparation}}  (\bibinfo{year}{2018}).

\bibitem{Wolf2012}
\bibinfo{author}{Wolf, M.}
\newblock \bibinfo{title}{Quantum channels \ operations: Guided tour.}
\newblock \emph{\bibinfo{journal}{Lecture Notes}}  (\bibinfo{year}{2012}).
\newblock
  \urlprefix\url{http://www-m5.ma.tum.de/foswiki/pub/M5/Allgemeines/\\MichaelWolf/QChannelLecture.pdf.}

\bibitem{Goodman2009}
\bibinfo{author}{Goodman, R.} \& \bibinfo{author}{Wallach, N.~R.}
\newblock \emph{\bibinfo{title}{{Symmetry, Representations, and Invariants}}}.
\newblock Graduate Texts in Mathematics (\bibinfo{publisher}{Springer},
  \bibinfo{year}{2009}).

\bibitem{Ruskai2002}
\bibinfo{author}{Ruskai, M.~B.}, \bibinfo{author}{Szarek, S.} \&
  \bibinfo{author}{Werner, E.}
\newblock \bibinfo{title}{An analysis of completely-positive trace-preserving
  maps on m2}.
\newblock \emph{\bibinfo{journal}{Lin. Alg. and its Appl.}}
  \textbf{\bibinfo{volume}{347}}, \bibinfo{pages}{159 -- 187}
  (\bibinfo{year}{2002}).

\bibitem{Nielsen2002}
\bibinfo{author}{Nielsen, M.~A.}
\newblock \bibinfo{title}{{A simple formula for the average gate fidelity of a
  quantum dynamical operation}}.
\newblock \emph{\bibinfo{journal}{Phys. Lett. A}}
  \textbf{\bibinfo{volume}{303}}, \bibinfo{pages}{249} (\bibinfo{year}{2002}).

\bibitem{Clifford2016}
\bibinfo{author}{Helsen, J.}, \bibinfo{author}{Wallman, J.~J.} \&
  \bibinfo{author}{Wehner, S.}
\newblock \bibinfo{title}{Representations of the multi-qubit clifford group}.
\newblock \emph{\bibinfo{journal}{Journ. Math. Phys.}}
  \textbf{\bibinfo{volume}{59}}, \bibinfo{pages}{072201}
  (\bibinfo{year}{2018}).

\bibitem{Dankert2009}
\bibinfo{author}{Dankert, C.}, \bibinfo{author}{Cleve, R.},
  \bibinfo{author}{Emerson, J.} \& \bibinfo{author}{Livine, E.}
\newblock \bibinfo{title}{{Exact and approximate unitary 2-designs and their
  application to fidelity estimation}}.
\newblock \emph{\bibinfo{journal}{Phys. Rev. A}} \textbf{\bibinfo{volume}{80}},
  \bibinfo{pages}{012304} (\bibinfo{year}{2009}).

\bibitem{corcoles2013process}
\bibinfo{author}{C{\'o}rcoles, A.~D.} \emph{et~al.}
\newblock \bibinfo{title}{Process verification of two-qubit quantum gates by
  randomized benchmarking}.
\newblock \emph{\bibinfo{journal}{Phys. Rev. A}} \textbf{\bibinfo{volume}{87}},
  \bibinfo{pages}{030301} (\bibinfo{year}{2013}).

\bibitem{titterington1985statistical}
\bibinfo{author}{Titterington, D.~M.}, \bibinfo{author}{Smith, A.~F.} \&
  \bibinfo{author}{Makov, U.~E.}
\newblock \emph{\bibinfo{title}{Statistical analysis of finite mixture
  distributions}} (\bibinfo{publisher}{Wiley}, \bibinfo{year}{1985}).

\end{thebibliography}
\clearpage
\clearpage

\appendix

\onecolumngrid
\section{Supplementary Methods I: Background material}
In this section we present, for the benefit of the reader, some well known facts about representation theory and the representation of quantum channels.  In particular we will review representations and characters and explain in more detail the Pauli transfer matrix formalism for quantum channels. More background on representations and characters can be found in \cite{Goodman2009,Fulton2004} while our presentation of quantum channels is based on \cite{Wolf2012,Nielsen2011}. \\
\subsection{I.1 Representation theory}\label{ssec:rep_theory}
We recall some useful facts about the representations of finite groups. For a more in depth treatment of this topic we refer to~\cite{Fulton2004,Goodman2009}.
Let $\gr{G}$ be a finite group and let $V$ be some finite dimensional complex vector space. Let also $\mathrm{U}(V)$ be the group of unitary linear transformations of $V$. We can define a \emph{representation} $\phi$ of the group $G$ on the space $V$ as a map
\begin{equation}\label{eq:representation}
\phi:\gr{G} \rightarrow \mathrm{U}(V): G\mapsto\phi(G)
\end{equation}
that has the property
\begin{equation}
\phi(G)\phi(H) = \phi(GH),\;\;\;\;\;\; \forall G,H \in \gr{G}.
\end{equation}
In general we will assume the operators $\phi(G)$ to be unitary.
If there is a non-trivial subspace $W$ of $V$ such that for all vectors $w\in W$ we have
\begin{equation}\label{eq:subrep}
\phi(G)w\in W,\;\;\;\;\;\;\forall G\in \gr{G},
\end{equation}
then the representation $\phi$ is called \emph{reducible}. The restriction of $\phi$ to the subspace $W$ is also a representation, which we call a \emph{subrepresentation} of $\phi$. If there are no non-trivial subspaces $W$  such that \cref{eq:subrep} holds the representation $\phi$ is called \emph{irreducible}.
Two representations $\phi,\phi'$ of a group $\gr{G}$ on spaces $V,V'$ are called \emph{equivalent} if there exists an invertible linear map $T:V\rightarrow V'$ such that
\begin{equation}\label{eq:equivalent}
T\phi(G) = \phi'(G)T,\;\;\;\;\;\;\;  \forall G\in \gr{G}.
\end{equation}
We will denote this by $\phi\equiv \phi'$.
For a representation $\phi$ on a vector space $V$ we can, for all linear maps $A:V\rightarrow V$ also define the \emph{twirl} of $A$ with respect to $\phi$. This is denoted as $\mc{T}_\phi$ and has the form
\begin{equation}
\mc{T}_{\phi}(A):= \frac{1}{|\gr{G}|}\sum_{G\in \gr{G}} \phi(G)A\phi(G)\ct.
\end{equation}
A general result called Maschke's lemma ensures that every representation of a group can be written as a direct sum of irreducible representations. That is we have for all representations $\phi$
\begin{align}
\phi(G) \simeq \bigoplus_{\lambda\in R_\gr{G}} \phi_\lambda(G)\tn{m_\lambda},\;\;\;\;\;\;\forall G\in \gr{G}
\end{align}
where the sum ranges over irreducible representations $\phi_\lambda$ of $\gr{G}$ and $m_\lambda$ is an integer denoting the multiplicity of $\phi_\lambda$ in $\phi$, that is, how many equivalent copies of the representation $\phi_\lambda$ are present in $\phi$. In this paper we will, for simplicity, mostly deal with representations $\phi$ that are \emph{multiplicity-free}. These are representations where $m_\lambda=1 $ for all $\phi_\lambda$.
The following corollary of Schur's lemma, an essential result from representation theory~\cite{Fulton2004,Goodman2009}, allows us to evaluate twirls over multiplicity-free representations.
\begin{lemma}[Lemma 1.7 and Prop. 1.8 in \cite{Fulton2004}]\label{lem:Schur}
Let $\gr{G}$ be a finite group and let $\phi$ be a multiplicity-free representation of $\gr{G}$ on a complex vector space $V$ with decomposition
\begin{align}
\phi(G) \simeq \bigoplus_{\lambda \in R_\gr{G}} \phi_\lambda(G),\;\;\;\;\;\;\forall G\in \gr{G}
\end{align}
into inequivalent irreducible subrepresentations $\phi_\lambda$.
Then for any linear map $A:V\to V$ the twirl of $A$ over $G$ takes the form
\begin{align}
\mc{T}_{\phi}(A) = \sum_{\lambda \in R_\gr{G}} \frac{\tr(AP_{\lambda})}{\tr(P_\lambda)} P_\lambda
\end{align}
where $P_\lambda$ is the projector onto the support of the representation $\phi_\lambda$.
\end{lemma}

We next recall the character of a representation. Let $\phi: \gr{G}\to V$
be a representation of a finite group $\gr{G}$ on a finite dimensional (real or
complex) vector space $V$. The character $\chi_\phi$ of
a representation $\phi$ is defined as
\begin{align}
\chi_{\phi}: \gr{G}\to\md{R}:G\mapsto  \chi_{\phi}(G) = \tr_V(\phi(G)),
\end{align}
where $\tr_V(\;)$ denotes the trace over the vector space $V$. Note that characters can in general be complex (that is, functions from $\gr{G}$ to $\md{C}$) but we will only consider representations with real valued characters here.
Characters have a number of useful properties~\cite{Fulton2004} which we recall here.
For representations $\phi,\phi'$ we have the relations
\begin{align}
\chi_{\phi\otimes \phi'} &= \chi_\phi\; \chi_{\phi'},\\
\chi_{\phi\oplus \phi'} &= \chi_\phi +\chi_{\phi'},
\end{align}
with suitable generalizations to multiple direct sums and tensor products. The following lemma, often referred to as the generalized projection formula, is of great use to us.
\begin{lemma}[Formula 2.32 in \cite{Fulton2004}]\label{lem:char_proj_form}
Let $\gr{G}$ be a group and let $\phi$  be a representation of $\gr{G}$.  Let also $\hat{\phi}$ be an irreducible subrepresentation of $\phi$ with associated character function $\chi_{\hat{\phi}}$. Then the following formula holds
\begin{equation}\label{eq:char_proj}
\frac{|\hat{\phi}|}{|\gr{G}|}\sum_{G\in \gr{G}}{\chi_{\hat{\phi}}(G)}\phi(G) = P_{\hat{\phi}},
\end{equation}
where $P_{\hat{\phi}}$ is the projector onto the support of all subrepresentations of $\phi$ that are equivalent to $\hat{\phi}$.
\end{lemma}
Note that in the presence of representations equivalent to $\hat{\phi}$, the projector on the RHS of \cref{eq:char_proj} projects onto  all subrepresentations that are equivalent to $\hat{\phi}$ rather than just $\hat{\phi}$.

\subsection{I.2 Pauli transfer matrix representation of quantum channels}\label{ssec:PTM}

Quantum channels~\cite{Wolf2012,Nielsen2011} are completely positive and trace-preserving (CPTP) linear maps $\mc{E}:\M \rightarrow \M$ where $\M$ is the Hilbert space of $2^q\times 2^q$ Hermitian matrices. 
We will denote quantum channels by calligraphic font throughout.
The canonical example of a quantum channel is conjugation by a unitary $U$, which we denote by the corresponding calligraphic letter, i.e. $\mc{U}(\rho) = U\rho U\ct$ for all density matrices $\rho$. We will denote the noisy implementation of a unitary channel by an overset tilde, e.g. $\widetilde{\mc{G}}$ denotes a noisy implementation some ideal unitary quantum channel channel $\mc{G}$.\\

It is often useful to think of quantum channels as matrices acting on vectors~\cite{Wolf2012,Wallman2014,Ruskai2002}. In order to do this we must choose a basis of the vector space $\M$. A convenient basis will be the basis of normalized Pauli matrices $\{\sigma_0\}\cup \bsq$ where 
$\sigma_0 := 2^{q/2}\id_{2^q}$ is the normalized identity matrix and 
\begin{equation}\label{pauli_norm}
\bsq :=\bigg\{2^{q/2}\{\id_2,X,Y,Z\}\tn{q}\bigg\}\backslash\{\sigma_0\},
\end{equation}
where $X,Y,Z$ are the standard single qubit Pauli matrices.
This set spans $\M$ and becomes an orthonormal basis when we equip $\M$ with the trace (or Hilbert-Schmidt) inner product defined as 
\begin{equation}
\inp{A}{B} := \tr(AB\ct),\;\;\;\;\;\; \forall A,B \in \M.
\end{equation}
For any element $A$ of $\M$ we will denote its vectorization as $\ket{A}$. $\ket{A}$ is a column vector of length $2^{2q}$ obtained by considering the set $\{\sigma_0\}\cup \boldsymbol{\sigma}_q$ as a basis for $\M$, that is 
\begin{equation}
\ket{A} = \sum_{\sigma \in \{\sigma_0\}\cup \bsq} \inp{A}{\sigma}\ket{\sigma}.
\end{equation}
$\ket{\;}$ has a natural dual which we denote by $\bra{\;}$.

As any quantum channel $\mc{E}$ is a linear map from $\M$ to itself we have
\begin{align}
|\mc{E}(\rho)\rangle = \sum_{\sigma \in\{\sigma_0\}\cup \bsq} |\mc{E}(\sigma)\rangle\!\rangle\!\langle\!\langle\sigma|\rho\rangle\!\rangle,
\end{align}
so that we can represent $\mc{E}$ by the matrix
\begin{align}
\mc{E} = \sum_{\sigma \in \{\sigma_0\}\cup\bsq} |\mc{E}(\sigma)\rangle\!\rangle\!\langle\!\langle \sigma|,
\end{align}
where we abuse notation by using the same symbol to refer to an abstract channel and its matrix representation. We will call this matrix the Pauli Transfer Matrix (PTM) representation of the channel $\mc{E}$.
The action of a channel $\mc{E}$ on a density matrix $\rho$ now corresponds to the standard matrix action on the vector $\ket{\rho}$, hence for a density matrix $\rho$ and a POVM element $Q$ in $\M$ we have
\begin{align}
\mc{E}|\rho\rangle &= \ket{\mc{E}(\rho)},\\
\tr(Q\mc{E}(\rho)) &= \bra{Q}\mc{E}\ket{\rho}. 
\end{align}

The PTM representation has the properties (as can be easily checked) that the composition of quantum channels is equivalent to matrix multiplication of their PTM representations and that tensor products of channels correspond to tensor products of the corresponding PTM representations, that is, for all channels $\mc{E}_1$ and $\mc{E}_2$ and all $A\in\M$,
\begin{align}
\ket{\mc{E}_1\circ\mc{E}_2(A)} &= \mc{E}_1\mc{E}_2\ket{A},\label{eq:composition} \\
\ket{\mc{E}_1\otimes \mc{E}_2(A\tn{2})} &= \mc{E}_1\otimes \mc{E}_2\ket{A\tn{2}}.
\end{align}

Another property of the PTM representation is that it is an actual representation (in the sense of \cref{eq:representation}) of any subgroup of the unitary group $U(2^q)$. This essentially follows from \cref{eq:composition}. For $U,V \in U(2^q)$ set $W = UV$. We then have for the PTM representation
\begin{equation}
\mc{U}\mc{V}\ket{X} =\ket{\mc{U}(\mc{V}(X))}= \ket{VUXU\ct V\ct} = \ket{VU X(VU)\ct} = \ket{WXW\ct}= \ket{\mc{W}(X)} = \mc{W}\ket{X},\;\;\;\;\;\;\;\forall X\in \M,
\end{equation}
which is essentially the definition of a representation.\\

\section{Supplementary Methods II: Standard randomized benchmarking with arbitrary finite groups}
In this section we give a quick overview of the standard randomized benchmarking procedure and how it applies to arbitrary finite groups. For a thorough exposition of randomized benchmarking with arbitrary finite groups, which also covers the case of groups with non-multiplicity-free PTM representations, see~\cite{francca2018approximate}. We will limit ourselves to gate-independent noise for ease of exposition. We begin by reviewing the randomized benchmarking procedure:
\begin{enumerate}
	\setlength\itemsep{-0.2em}
	\item Choose a state $\rho$ and a two-component POVM $\{Q, \id -Q\}$ such that $\tr(Q\rho)$ is large
	\item Sample $ \vec{G}  =G_1,\ldots, G_m$ uniformly at random from $\gr{G}$  
	\item Prepare the state $\rho$ and apply the gates $G_1,G_2,\ldots G_m$ 
	\item Compute the inverse $G_{\mathrm{inv}} = (G_m\cdots G_1)\ct$ and apply it to $\rho$
	\item Estimate the survival probability $p_m(\vec{G}) :=\bra{Q}\widetilde{\mc{G}}_{\mathrm{inv}}\widetilde{\mc{G}}_m\cdots \widetilde{\mc{G}}_1\ket{\rho}$
	\item Repeat steps 3-8 for many $\vec{G}$ and estimate the average $p_m := \md{E}_{\vec{G}}(p_m(\vec{G}))$
	\item Repeat steps 3-9 for all sequence lengths $m\in \md{M}$ (where $\md{M}$ is some pre-chosen set)
	\item Output $\{p_m\}_{m\in \md{M}}$
\end{enumerate}
We now give justification for \cref{eq:rand_bench_av,eq:rand_bench_av_proj} in the main text. We have the following lemma.
\begin{lemma}[\cref{eq:rand_bench_av}]\label{lem:rand_bench_av}
Let $\gr{G}$ be a finite subgroup of $U(2^q)$ such that the PTM representation $\mc{G}=\oplus_{\lambda\in R_{\gr{G}}}\phi_\lambda(G)$ is multiplicity-free.  Let $\tilde{\mc{G}} = \mc{E} \mc{G}$ be some implementation of the operation $G\in \gr{G}$ with $\mc{E}$ a CPTP map. Consider the average survival probability $p_m$ of a randomized benchmarking experiment of sequence length $m$ with an input state $\rho$ and an output two-component POVM $\{Q, \id-Q\}$,
\begin{equation}
p_m = \avg_{G_1, \ldots G_m}\bra Q\widetilde{\mc{G}}_\mathrm{inv} \widetilde{\mc{G}}_m \cdots \widetilde{\mc{G}}_1\ket{\rho}.
\end{equation}
We now have that 
\begin{equation}
p_m = \bra{Q} \left(\avg_{G\in \gr{G}}\mc{G}\ct \mc{E}\mc{G}\right)^m\ket{\rho} = \sum_{\lambda\in R_{\gr{G}}}\bra{Q}\mc{P}_\lambda\ket{\rho} f_{\lambda}^m.
\end{equation}
\end{lemma}
\begin{proof}
We begin by noting that $\tilde{\mc{G}}_{\mathrm{inv}}=\mc{E}\mc{G}_1\ct \cdots\mc{G}_m\ct $ Using this and the fact that $\tilde{\mc{G}} = \mc{E} \mc{G}$ for all $G\in \gr{G}$ we can write
\begin{align}
p_m &= \avg_{G_1, \ldots G_m}\bra Q\widetilde{\mc{G}}_\mathrm{inv} \widetilde{\mc{G}}_m \cdots \widetilde{\mc{G}}_1\ket{\rho}\\
&= \bra{Q} \avg_{G_1, \ldots G_m} \mc{E}\mc{G}_1\ct \cdots \mc{G}_m\ct \mc{E}\mc{G}_m \mc{E}\mc{G}_{m-1}\cdots \mc{E}\mc{G}_1\ket{\rho}.
\end{align}
Noting that the operator $\md{E}_{G_m\in \gr{G}}\mc{G}_m\ct \mc{E}\mc{G}_m$ commutes with $\mc{G}$ for all $G\in \gr{G}$ we can write
\begin{align}
p_m  &= \bra{Q} \avg_{G_1, \ldots G_{m-1}} \mc{E}\mc{G}_1\ct \cdots\mc{G}_{m-1}\left(\avg_{G_m\in \gr{G}} \mc{G}_m\ct \mc{E}\mc{G}_m\right) \mc{E}\mc{G}_{m-1}\cdots \mc{E}\mc{G}_1\ket{\rho}\\
&= \bra{Q}\avg_{G_1,\ldots,G_{m-2}}\mc{E}\mc{G}_1\ct \cdots\mc{G}_{m-2}\left(\avg_{G_m\in \gr{G}} \mc{G}_m\ct \mc{E}\mc{G}_m\right) \left(\avg_{G_{m-1}\in \gr{G}} \mc{G}_{m-1}\ct \mc{E}\mc{G}_{m-1}\right)\mc{E}\mc{G}_{m-2}\cdots \mc{E}\mc{G}_1\ket{\rho}.
\end{align}
Repeating this procedure we obtain
\begin{equation}
p_m = \bra{Q}\left(\avg_{G\in \gr{G}} \mc{G}\ct \mc{E}\mc{G}\right)^m\ket{\rho}
\end{equation}
Where we have set $Q\rightarrow \mc{E}\ct(Q)$. Now we use Schur's lemma (\cref{lem:Schur}) and the fact that $\mc{G}= \bigoplus_{\lambda\in R_{\gr{G}}}\phi_{\lambda}(G)$ to obtain
\begin{align}
p_m &= \bra{Q}\left(\sum_{\lambda\in R_{\gr{G}}} f_\lambda \mc{P}_\lambda\right)^m\ket{\rho}\\
&=\sum_{\lambda\in R_{\gr{G}}}f_\lambda^m \bra{Q}\mc{P}_\lambda \ket{\rho}
\end{align}
where we have set $f_{\lambda}: = \tr(\mc{P}_\lambda\mc{E})/\tr(\mc{P}_\lambda)$. This completes the proof.
\end{proof}

\section{Supplementary Methods III: Average fidelity and quality parameters}
In this section we discuss the relation of the average fidelity of a quantum channel to the quality parameters $f_\lambda$ generated by character randomized benchmarking, under the assumption of gate-independent noise.
We begin by recalling the definition of the average fidelity (to the identity) of a quantum channel $\mc{E}$.
\begin{definition}
Let $\mc{E}$ be a quantum channel. Its average fidelity (with respect to the identity channel) $F_{\mathrm{avg}}(\mc{E})$ is defined as
\begin{equation}
F_{\mathrm{avg}}(\mc{E}) := \int d\psi \tr(\dens{\psi}\mc{E}(\dens{\psi})),
\end{equation}
where $\dens{\psi}$ is the regular density matrix of the pure state $\psi$ and the integral is taken over the Haar measure on the set of pure states.
\end{definition}
The average fidelity  $F_{\mathrm{avg}}(\mc{E})$ of a quantum channel $\mc{E}$ is related to the trace (taken over superoperators) of the Pauli transfer matrix of $\mc{E}$. We have the following lemma.

\begin{lemma}\label{lem:trace_to_fid}
Let $\mc{E}$ be a CPTP map acting on a system of $q$ qubits. We have that 
\begin{equation}
F_{\mathrm{avg}}(\mc{E}) = \frac{2^{-q}\tr(\mc{E})+1}{2^q+1}
\end{equation}
\end{lemma}
\begin{proof}
Note that $F_{\mathrm{avg}}(\mc{E})$ is invariant under unitary conjugation~\cite{Nielsen2002}, that is $F_{\mathrm{avg}}(\mc{E}) = F_{\mathrm{avg}}(\mc{U}\ct\mc{E}\mc{U})$ for all $U\in U(2^q)$. Similarly we have, by cyclicity of the trace that $\tr(\mc{E}) = \tr(\mc{U}\ct\mc{E}\mc{U})$. Because both the trace and $F_{\mathrm{avg}}$ are linear we moreover have that
\begin{align}
F_{\mathrm{avg}}(\mc{E}) = F_{\mathrm{avg}}\left(\int dU \mc{U}\ct\mc{E}\mc{U}\right),\\
\tr(\mc{E}) = \tr\left(\int dU \mc{U}\ct\mc{E}\mc{U}\right).
\end{align}
From \cite{Nielsen2002} it is known that there exist a $p\in [-1/(2^{2q}-1),1]$ such that
\begin{equation}
\int dU \mc{U}\ct\mc{E}\mc{U}\ket{\rho} = p\ket{X} + \frac{1-p}{d}\tr(X)\ket{\id}
\end{equation}
for all operators $X\in \M$, i.e. $\int dU \mc{U}\ct\mc{E}\mc{U}$ is a depolarizing channel. Evaluating the average fidelity we get
\begin{equation}
F_{\mathrm{avg}}(\mc{E}) = F_{\mathrm{avg}}\left(\int dU \mc{U}\ct\mc{E}\mc{U}\right) = p + \frac{1-p}{2^q}
\end{equation}
and similarly evaluating the trace we get 
\begin{equation}
\tr(\mc{E}) = \tr\left(\int dU \mc{U}\ct\mc{E}\mc{U}\right) = 1 + (2^{2q}-1)p
\end{equation}
from which the lemma follows.
\end{proof}

In the context of character randomized benchmarking, if $\gr{G}$ is a group with implementation $\mc{E}\mc{G}$ (for all $\mc{G}$) we can relate the average fidelity $F_{\mathrm{avg}}(\mc{E})$ of the quantum channel $\mc{E}$ to the quality parameters $f_\lambda,\; \lambda\in R_\gr{G}$ generated by the character randomized benchmarking experiment. More precisely we have the following lemma which shows that the average fidelity can be related to a weighted average of the quality parameters.

\begin{lemma}\label{lem:trace_to_quality}
Let $\gr{G}$ be a subgroup of $U(2^q)$ such that the PTM representation $\mc{G}= \bigoplus_{\lambda\in R_{\gr{G}}}\phi_\lambda(G)$ for $\lambda \in R_\gr{G}$ is multiplicity-free. We have for any quantum channel $\mc{E}$ that the twirl of $\mc{E}$ with respect to $\mc{G}$ is of the form
 \begin{equation}\label{eq:thm_schur_decomp}
\avg_{G\in \gr{G}} \mc{G}\ct\mc{E}\mc{G} = \sum_{\lambda\in R_\gr{G}} f_\lambda \mc{P}_\lambda
\end{equation}
where $f_\lambda = \tr(\mc{P}_\lambda\mc{E})/\tr(\mc{P}_\lambda)$ and $\mc{P}_\lambda$ is the projection onto the support of the representation $\phi_\lambda$. Moreover the average fidelity $F_{\mathrm{avg}}(\mc{E})$ of $\mc{E}$ is given by
\begin{equation}
F_{\mathrm{avg}}(\mc{E}) = \frac{2^{-q}\sum_{\lambda\in R_\gr{G}} f_\lambda \mc{P}_\lambda+1}{2^q+1}
\end{equation}
\end{lemma}
\begin{proof}
\Cref{eq:thm_schur_decomp} follows from a standard application of Schur's lemma (\cref{lem:Schur}). Now consider the trace of $\mc{E}$, from \cref{eq:thm_schur_decomp} and the linearity and cyclicity of the trace we have that
\begin{equation}
\tr(\mc{E}) = \tr\left(\avg_{G\in \gr{G}} \mc{G}\ct\mc{E}\mc{G}\right) = \sum_{\lambda \in R_\gr{G}}\tr(\mc{P}_\lambda)f_\lambda.
\end{equation}
Using \cref{lem:trace_to_fid} we obtain the lemma statement.
\end{proof}

\section{Supplementary Methods IV: Character randomized benchmarking}\label{sec:char_randomized benchmarking}

In this section we will more formally write down the central results of the main text. We will give an analysis of character randomized benchmarking in the case of gate independent noise (which is a formalization of the results given in the main text) and an analysis of character randomized benchmarking in the case of gate-dependent noise. This last part is significantly more technical than the first two. We begin by formally writing down what we mean by a `character randomized benchmarking experiment'
\begin{definition}\label{def:char_randomized benchmarking}
A character randomized benchmarking experiment is defined by a tuple $(\gr{G}, \hat{\gr{G}}, \lambda',\md{M})$ where $\gr{G}$ is a group such that the PTM representation $\mc{G}=\oplus_{\lambda\in R_\gr{G}}\phi_\lambda(G)$ is multiplicity-free, $\gr{\hat{G}}$ is a subgroup of $\gr{G}$, $\lambda'\in R_\gr{G}$ is an element of the index set $R_{\gr{G}}$ labeling the irreducible subrepresentations of $\mc{G}$ and $\md{M}$ is a set of integers denoting the sequence lengths. A character randomized benchmarking experiment outputs a list of real numbers $\{k_m^{\lambda'}\}_{m\in \md{M}}$ given by the following procedure
\begin{enumerate}
	\setlength\itemsep{-0.2em}
	\item Choose an irreducible subrepresentation $\hat{\phi}$ of the PTM representation $\mc{\hat{G}}$ of $\gr{G}$ such that $\mc{P}_{\hat{\phi}}\mc{P}_{\lambda'} = \mc{P}_{\hat{\phi}}$.
	\item Choose a state $\rho$ and a two-component POVM $\{Q, \id -Q\}$ such that $\tr(Q\mc{P}_{\hat{\phi}}(\rho))$ is maximized
	\item Sample $ \vec{G}  =G_1,\ldots, G_m$ uniformly at random from $\gr{G}$ 
	\item Sample $\hat{G}$ uniformly at random from $\hat{\gr{G}}$ 
	\item Prepare the state $\rho$ and apply the gates $(G_1\hat{G}),G_2,\ldots G_m$ (note that we compile $G_1,\hat{G}$ into a single gate)
	\item Compute the inverse $G_{\mathrm{inv}} = (G_m\cdots G_1)\ct$ and apply it to $\rho$ (note that $\hat{G}$ is not inverted)
	\item Estimate the weighted `survival probability' $k^{\lambda}_m(\vec{G},\hat{G}) :=|\hat{\phi}|\chi_{\hat{\phi}}(\hat{G})\bra{Q}\mc{\widetilde{G}}_{\mathrm{inv}}\mc{\widetilde{G}}_m\cdots \widetilde{(\mc{G}_1\mc{\hat{G}})}\ket{\rho}$ with $\chi_{\hat{\phi}}$ the character function of $\hat{\phi}$
	\item Repeat steps 3-7 for many $\hat{G}\in \hat{\gr{G}}$ and estimate the average $k^{\lambda'}_m(\vec{G}) := \md{E}_{\hat{G}}(k^{\lambda'}_m(\vec{G},\hat{G}))$
	\item Repeat steps 3-8 for many $\vec{G}$ and estimate the average $k_m^{\lambda'} := \md{E}_{\vec{G}}(k^{\lambda'}_m(\vec{G}))$
	\item Repeat steps 3-9 for all $m\in \md{M}$
	\item Output $\{k^{\lambda'}_m\}_{m\in \md{M}}$
\end{enumerate}

\end{definition}
The set of numbers $\{k_m^{\lambda'}\}_{m\in \md{M}}$ can then be fitted to an exponential decay, to extract the quality parameter $f_{\lambda'}$  Given a group $\gr{G}$ we can perform character randomized benchmarking experiments for each $\lambda'\in R_{\gr{G}}$ obtaining a list $\{f_{\lambda'}\;\|\;\lambda\in R_{\gr{G}}\}$. This list of quality parameters can be associated to the average fidelity of the gateset $\gr{G}$ using \cref{lem:trace_to_quality,lem:trace_to_fid}.
For completeness we also give an interleaved version of the character randomized benchmarking protocol. 
\begin{definition}\label{def:int_char_randomized benchmarking}
An interleaved  character randomized benchmarking experiment is defined by a tuple $(\gr{G}, \hat{\gr{G}}, \lambda',\md{M},C)$ where $\gr{G}$ is a group such that the PTM representation $\mc{G}=\oplus_{\sigma\in R_\gr{G}}\phi_\sigma(G)$ is multiplicity-free, $\gr{\hat{G}}$ is a subgroup of $\gr{G}$, $\lambda'\in R_\gr{G}$ is an element of the index set $R_{\gr{G}}$ labeling the irreducible subrepresentations of $\mc{G}$, $\md{M}$ is a set of integers denoting the sequence lengths and $C$ is a quantum gate such that $\langle\gr{G},C\rangle$ is a finite group. An interleaved character randomized benchmarking experiment outputs a list of real numbers $\{k_m^{\lambda'}\}_{m\in \md{M}}$ given by the following procedure
\begin{enumerate}
	\setlength\itemsep{-0.2em}
	\item Choose an irreducible subrepresentation $\hat{\phi}$ of the PTM representation $\mc{\hat{G}}$ of $\gr{G}$ such that $\mc{P}_{\hat{\phi}}\mc{P}_{\lambda'} = \mc{P}_{\hat{\phi}}$.
	\item Choose a state $\rho$ and a two-component POVM $\{Q, \id -Q\}$ such that $\tr(Q\mc{P}_{\hat{\phi}}(\rho))$ is maximized
	\item Sample $ \vec{G}  =G_1,\ldots, G_m$ uniformly at random from $\gr{G}$ 
	\item Sample $\hat{G}$ uniformly at random from $\hat{\gr{G}}$ 
	\item Prepare the state $\rho$ and apply the gates $(G_1\hat{G}),C, G_2,C, \ldots C,G_m,C$ (note that we compile $G_1,\hat{G}$ into a single gate)
	\item Compute the inverse $G_{\mathrm{inv}} = (CG_m C\cdots C G_1)\ct$ and apply it to $\rho$ (note that $\hat{G}$ is not inverted, but $C$ is)
	\item Estimate the weighted `survival probability' $k^{\lambda}_m(\vec{G},\hat{G}) :=|\hat{\phi}|\chi_{\hat{\phi}}(\hat{G})\bra{Q}\mc{\widetilde{G}}_{\mathrm{inv}}\mc{\widetilde{G}}_m\cdots \widetilde{(\mc{G}_1\mc{\hat{G}})}\ket{\rho}$ with $\chi_{\hat{\phi}}$ the character function of $\hat{\phi}$
	\item Repeat steps 3-7 for many $\hat{G}\in \hat{\gr{G}}$ and estimate the average $k^{\lambda'}_m(\vec{G}) := \md{E}_{\hat{G}}(k^{\lambda'}_m(\vec{G},\hat{G}))$
	\item Repeat steps 3-8 for many $\vec{G}$ and estimate the average $k_m^{\lambda'} := \md{E}_{\vec{G}}(k^{\lambda'}_m(\vec{G}))$
	\item Repeat steps 3-9 for all $m\in \md{M}$
	\item Output $\{k^{\lambda'}_m\}_{m\in \md{M}}$
\end{enumerate}

\end{definition}

\section{Supplementary Methods V: Examples of character randomized benchmarking}
In this section we give a more detailed overview of the two examples given in the text; benchmarking a gateset with a $T$-gate and 2-for-1 interleaved benchmarking. We begin with an exposition of the irreducible representations of the PTM representation of the the Pauli group, as this is the choice for $\gr{\hat{G}}$ in both examples.

\subsection{V.1 Representations of the Pauli group}\label{ssec:rep_pauli}
Probably the most useful choice for the group $\gr{\hat{G}}$ is the multi-qubit Pauli group. This group is defined as  $\gr{P}_q = \langle i\id, X, Z\rangle\tn{q}$. The reason this group is useful lies in the fact that the irreducible subrepresentations of the Pauli transfer matrix representations of $\gr{P}_q$ are all of dimension one and moreover that they are all inequivalent. We have the following lemma
\begin{lemma}\label{lem:rep_pauli}
Let $\gr{P}_q$ be the Pauli group on $q$ qubits and consider its PTM representation. The PTM representation decomposes as
\begin{equation}
\mc{P} = \bigoplus_{\sigma \in \{\sigma_0\}\cup\bf{\sigma}_q} \phi_\sigma(P),\;\;\;\;\;\;\;\forall P\in \gr{P}_q
\end{equation}
with the projector $\mc{P}_\sigma$ onto the support of $\phi_\sigma$ given by
\begin{equation}
\mc{P}_\sigma = \ket{\sigma}\!\bra{\sigma}
\end{equation}
for all $\sigma \in\{\sigma_0\}\cup\bf{\sigma}_q $. Moreover all representations $\phi_\sigma$ are one-dimensional, mutually inequivalent and have character functions $\chi_\sigma$ given by
\begin{equation}
\chi_{\sigma}(P) = (-1)^{\inp{\sigma}{P}}
\end{equation}
with 
\begin{equation}
\inp{\sigma}{P} = \begin{cases} 0  \iff \text{$P$ and $\sigma$ commute}\\1  \iff \text{$P$ and $\sigma$ anti-commute}. 
\end{cases}
\end{equation}
\end{lemma}
\begin{proof}
Consider the action of $\mc{P}$ on the vector $\ket{\sigma}$ for $\sigma \in\{\sigma_0\}\cup\bf{\sigma}_q $ and $P\in \gr{P}_q$:
\begin{equation}
\mc{P}\ket{\sigma} = \ket{P\sigma P\ct} =  (-1)^{\inp{\sigma}{P}}\ket{PP\ct \sigma} = (-1)^{\inp{\sigma}{P}} \ket{\sigma}.
\end{equation}
This means that $\ket{\sigma}$ spans a subrepresentation of $\mc{P}$. Since the space spanned by $\ket{\sigma}$ is one dimensional, this subrepresentation is also irreducible. We call this subrepresentation $\phi_\sigma$. By construction $\mc{P}_\sigma = \ket{\sigma}\bra{\sigma}$. Moreover the character function $\chi_\sigma$ is given as
\begin{equation}
\chi_{\sigma}(P) = \tr(\mc{P} \ket{\sigma}\bra{\sigma}) =\bra{\sigma}\mc{P}\ket{\sigma} = (-1)^{\inp{\sigma}{P}}
\end{equation}
It remains to prove that for $\sigma \neq \sigma'$ the representations $\sigma, \sigma'$ are inequivalent. We do this by leveraging the following fundamental result from character theory.
% \begin[\cite[Equation 2.10]{Fulton2004}]{lemma}
%  Let $\phi,\phi'$ be irreducible representations of a group $\gr{G}$ with characters $\chi_\phi,\chi_{\phi'}$. Then $\phi, \phi'$ are inequivalent if and only if the character inner product
%  \begin{equation}
%  \inp{\chi_{\phi}}{\chi_{\phi'}} = \avg_{G\in \gr{G}} \chi_{\phi}(G)\bar{\chi}_{\phi'}(G)
%  \end{equation}
%  is equal to zero.
%  \end{lemma}
We calculate the character inner product for representations $\phi_\sigma,\phi_{\sigma'}$ of $\gr{P}_q$
 as follows:
\begin{equation}
\inp{\chi_{\sigma}}{\chi_{\sigma'}} = \avg_{P\in \gr{P}_q} \chi_{\sigma}(P)\bar{\chi}_{\sigma'}(P) = \avg_{P\in \gr{P}_q} (-1)^{\inp{P}{\sigma}}(-1)^{\inp{P}{\sigma'}}.
\end{equation}
It is easy to verify by explicit computation that $(-1)^{\inp{P}{\sigma}}(-1)^{\inp{P}{\sigma'}} = (-1)^{\inp{P}{\tau}}$ with $\tau \approx \sigma \sigma'$, i.e $\tau$ is equal to $\sigma \sigma'$ up to a proportionality factor. Since $\sigma \sigma'\approx \id$ if and only if $\sigma = \sigma'$ we have that $\tau \neq \id$. Since a non-identity Pauli matrix (such as $\tau$) commutes with precisely half of the elements of the Pauli group and anti-commutes with the other half (for a proof of this fact see for instance~\cite[Lemma 1]{Clifford2016}) we have that $\inp{\chi_{\sigma}}{\chi_{\sigma'}}=0 $, completing the lemma.
\end{proof}

Note that for two Pauli matrices $P, P'$ we can also efficiently (in the number of qubits $q$) decide whether they commute or anti-commute. This means that the character function $\chi_\sigma(P)$ can be efficiently computed on the fly for any $\sigma$ and $P$. This is important because we must compute an instantiation of the character function for every random sample drawn during the character randomized benchmarking procedure. Note however that this can be done in post-processing so high speed (not just efficient) calculation of the character function is not a requirement for the success of the character randomized benchmarking procedure.

\subsection{V.2 Benchmarking a $T$ gate}
In this section we give some more background information on how to perform character randomized benchmarking on the CNOT-dihedral gateset $\gr{T}_q$ which is defined as all gates that can be synthesized from a combination of $T$ gates, $X$ gates and CNOT gates, or more formally
\begin{equation}
\gr{T}_q = \langle \CNOT_{i,j}, T_k, X_l \;\;\|\;\; i,j, k, l\in\{1, \ldots,q\}, \;i\neq j\rangle
\end{equation}
where $\CNOT_{i,j}$ indicates the $\CNOT$ gate with the $i$'th qubit as control and the $j$'th qubit as target, $T_k$ indicates the $T$-gate applied to the $k$'th qubit and $X_l$ indicates the $X$-gate applied to the $l$'th qubit.
 The PTM representation of this group has, as mentioned in the main text, three irreducible subrepresentations $\phi_1,\phi_2\phi_3$, with associated projections:
\begin{align}
\mc{P}_1 &= \ddens{\sigma_0}\\
\mc{P}_2 &= \sum_{\sigma \in \mc{Z}} \ddens{\sigma}\\
\mc{P}_3 &= \sum_{\sigma \in \bf{\sigma}_q\backslash \mc{Z}} \ddens{\sigma}
\end{align}
where $\mc{Z}$ is defined as the subset of normalized Pauli's consisting of only $Z$ and $\id$ tensor factors. The above was proven in \cite{Cross_2016}. Since there are three representations, we must estimate three quality parameters $f_1,f_2,f_3$ in order to estimate the average fidelity. However, assuming the noisy gates are CPTP maps it is easy to see that $f_1=1$. This leaves us with estimating the parameters $f_2,f_3$. This we do by two character randomized benchmarking experiments which we describe explicitly below.\\

\noindent {\bf Estimating $f_2$}\\

To estimate the quality parameter $f_2$ we must perform the following set of steps

\begin{enumerate}
\item Choose $\gr{G} = \gr{T}_q$ the CNOT-dihedral group on $q$ qubits and choose $\gr{\hat{G}} = \gr{P}_q$ the $q$ qubit Pauli group
\item Choose $\{Q,\id- Q\}$ a two component POVM with $Q= \frac{1}{2}(\id+ Z\tn{q})$ and choose $\rho= 2^{-q}(\id+ Z\tn{q})$ (see section VI on how to prepare this non-pure state efficiently)
\item Choose $\phi_{\sigma}$ with $\sigma =2^{q/2}Z\tn{q}$ an irreducible subrepresentation of the PTM representation of $\gr{P}_q$ with character function $\chi_{\sigma}$ (which can be computed from \cref{lem:rep_pauli})
\item Perform a character randomized benchmarking experiment (as given in \cref{def:char_randomized benchmarking}) $(\gr{T}_q,\gr{P}_q, 2, \md{M})$ (for suitably chosen $\md{M}$) with $\phi = \phi_\sigma$ to obtain the quality parameter $f_2$.
\end{enumerate}

\noindent {\bf Estimating $f_3$}\\

To estimate the quality parameter $f_3$ we must perform the following set of steps

\begin{enumerate}
\item Choose $\gr{G} = \gr{T}_q$ the CNOT-dihedral group on $q$ qubits and choose $\gr{\hat{G}} = \gr{P}_q$ the $q$ qubit Pauli group
\item Choose $\{Q,\id- Q\}$ a two component POVM with $Q= \frac{1}{2}(\id+ X\tn{q})$ and choose $\rho= 2^{-q}(\id+ X\tn{q})$ (see section VI on how to prepare this non-pure state efficiently)
\item Choose $\phi_{\sigma}$ with $\sigma =2^{q/2} X\tn{q}$ an irreducible subrepresentation of the PTM representation of $\gr{P}_q$ with character function $\chi_{\sigma}$ (which can be computed from \cref{lem:rep_pauli})
\item Perform a character randomized benchmarking experiment (as given in \cref{def:char_randomized benchmarking}) $(\gr{T}_q,\gr{P}_q, 3, \md{M})$ (for suitably chosen $\md{M}$) with $\phi = \phi_\sigma$ to obtain the quality parameter $f_3$.
\end{enumerate}

\noindent {\bf Computing the average fidelity}\\

The average fidelity can now be computed from \cref{lem:trace_to_quality} and \cref{lem:trace_to_fid}, provided we know the quantities $\tr(\mc{P}_2)$ and $\tr(\mc{P}_3)$. These were derived in \cite{Cross_2016} giving an average fidelity formula of the form
\begin{equation}
F_{\mathrm{avg}} = \frac{2^q-1}{2^q}\left(1 - \frac{f_2 + 2^q f_3}{2^q+1}\right)
\end{equation}

\subsection{V.3 2-for-1 interleaved benchmarking}\label{ssec:2for1_benchmarking}
In this section we give some more detailed information on 2-for-1 interleaved benchmarking. The aim of this section is two-fold: (1) gather all information needed to perform 2-for-1 randomized benchmarking in one place and (2) detail a simulation showcasing the benefits of 2-for-1 randomized benchmarking. The goal of this protocol is to extract the average fidelity associated to a single two qubit gate $C$. This is usually done using interleaved randomized benchmarking on the $2$-qubit Clifford group. Here we will replace this $2$-qubit Clifford group by two copies of the single qubit Clifford group. We begin by analyzing the behavior of character randomized benchmarking using $\gr{C}_1\tn{2}$. We have the following lemma, which justifies \cref{eq:2for1_reps} in the main text.
\begin{lemma}
Let $\gr{G}=\gr{C}_1\tn{2}$ be the two-fold tensor product of the single qubit Clifford group. The PTM representation of this group (acting on two qubits), decomposes into four inequivalent irreducible subrepresentations $\phi_w$ indexed by $w\in \{0,1\}^{\times 2}$ with projectors onto the supports of $\phi_w$ given by
\begin{align}
\mc{P}_{(0,0)} = \ddens{\sigma_0\otimes\sigma_0}\\
\mc{P}_{(1,0)} = \sum_{\sigma\in \bf{\sigma}_1 }\ddens{\sigma\otimes\sigma_0}\\
\mc{P}_{(0,1)} = \sum_{\sigma \in \bf{\sigma}_1}\ddens{\sigma_0\otimes\sigma}\\
\mc{P}_{(1,1)} = \sum_{\sigma,\sigma'\in \bf{\sigma}_1 }\ddens{\sigma\otimes\sigma'}
\end{align}
\end{lemma}
\begin{proof}
We begin by noting that for all $G\in \gr{C}_1$ we have that $\mc{G}\ket{\sigma_0} = \ket{\sigma_0}$. This already implies that
\begin{equation}
\mc{C} \mc{P}_w = \mc{P}_w \mc{C},\;\;\;\;\;\;\;\;C\in \gr{C}_1\tn{2},\;\; w\in \{0,1\}^{\times 2}
\end{equation}
which means all $\phi_w$ defined in the lemma statement are subrepresentations of the PTM representation of $\gr{C}_1\tn{2}$. To see that they are also irreducible we calculate the character inner product of the PTM representation of $\gr{C}_1\tn{2}$. We have
\begin{equation}
\inp{\chi_{\mathrm{PTM}}}{\chi_{\mathrm{PTM}}} = \avg_{C_1,C_2\in \gr{C}_1} |\tr(\mc{C}_1\otimes\mc{C}_2)|^2 = \left(\avg_{C_1\in \gr{C}_1} |\tr(\mc{C}_1)|^2\right)^2.
\end{equation}
Because the single qubit Clifford group is a two-design we know that $\md{E}_{C_1\in \gr{C}_1} |\tr(\mc{C}_1)|^2  =2$~\cite{Dankert2009} and hence that $\inp{\chi_{\mathrm{PTM}}}{\chi_{\mathrm{PTM}}} =4$. Since characters are additive w.r.t. taking direct sums of representations and $\inp{\chi_{\phi}}{\chi_{\phi}}\geq 1$ with equality if and only if $\phi$ is irreducible we conclude that $\phi_w$ must also be irreducible for all $w\in \{0,1\}^{\times 2}$.

\end{proof}

The 2-for-1 interleaved benchmarking protocol consists of two parts; the reference experiment and the interleaved experiment. We now list the steps required to perform 2-for-1 interleaved benchmarking, making all aspects of it (such as character functions) explicit.\\

\noindent{\bf Reference experiment}\\

To perform the reference stage of two-for-one interleaved benchmarking we must perform the following set of steps

\begin{enumerate}
\item Choose $\gr{G} = \gr{C}_1\tn{2}$ the group of single qubit Cliffords on two qubits and choose $\gr{\hat{G}} = \gr{P}_2$ the two qubit Pauli group
\item Choose $\{Q,\id- Q\}$ a two component POVM with $Q= \dens{00}$ and choose $\rho = \dens{00}$
\item Choose $\phi_{\sigma}$ with $\sigma = (Z\otimes \id )/2$ an irreducible subrepresentation of the PTM representation of $\gr{P}_2$ with character function $\chi_{\sigma}$ (given explicitly in \cref{box:char})
\item Perform a character randomized benchmarking experiment (as given in \cref{def:char_randomized benchmarking}) $(\gr{C}_1\tn{2},\gr{P}_2, w, \md{M})$ (for suitably chosen $\md{M}$) with $\phi = \phi_\sigma$ to obtain the quality parameter $f_w$ with $w= (1,0)$
\item Choose $\phi_{\sigma}$ with $\sigma = (\id\otimes Z )/2$ an irreducible subrepresentation of the PTM representation of $\gr{P}_2$ with character function $\chi_{\sigma}$ (given explicitly in \cref{box:char})
\item Perform a character randomized benchmarking experiment $(\gr{C}_1\tn{2},\gr{P}_2, w, \md{M})$ (for suitably chosen $\md{M}$) with $\phi = \phi_\sigma$ to obtain the quality parameter $f_w$ with $w= (0,1)$
\item Choose $\phi_{\sigma}$ with $\sigma = (Z\otimes Z )/2$ an irreducible subrepresentation of the PTM representation of $\gr{P}_2$ with character function $\chi_{\sigma}$ (given explicitly in \cref{box:char})
\item Perform a character randomized benchmarking experiment $(\gr{C}_1\tn{2},\gr{P}_2, w, \md{M})$ (for suitably chosen $\md{M}$) with $\phi = \phi_\sigma$ to obtain the quality parameter $f_w$ with $w= (1,1)$
\end{enumerate}
Knowing that $f_w = 1$ for $w=(0,0)$ (assuming the noise affecting the gates is CPTP) we can use \cref{lem:trace_to_quality,lem:trace_to_fid} to obtain the average reference fidelity $F_{\mathrm{avg}}^{\mathrm{{ref}}}$ as
\begin{equation}
F_{\mathrm{avg}}^{\mathrm{{ref}}} = \frac{1}{5}\left(\frac{1}{4}\left(1 + 3 f_{(0,1)} + 3f_{(1,0)} + 9 f_{(1,1)}\right) + 1\right).
\end{equation}
The character functions $\chi_\sigma$ for $\sigma \in \{(Z\otimes \id )/2,(\id\otimes Z )/2,(Z\otimes Z )/2\}$ are given in \cref{box:char}.\\
\begin{table}
\begin{tabular}{|c|c|c|c|c|c|c|c|c|c|c|c|c|c|c|c|c|}
\hline
$\sigma\backslash P$ & $\id\id$ & $Z\id$ & $\id Z$ & $ZZ$ & $X\id$ & $\id X$ & $XX$ & $Y\id $ & $\id Y$ & $YY$ & $ZX$ & $XZ$ & $ZY$ & $YZ$ & $XY$ & $YX$\\
\hline\hline
$Z\id$ &  $1$ & $1$ & $1$ & $1$ & $-1$ & $1$ & $-1$ & $-1$ & $1$ & $-1$ & $1$ & $-1$ & $1$ & $-1$ & $-1$ & $-1$\\
\hline
$\id Z$ &  $1$ & $1$ & $1$ & $1$ & $1$ & $-1$ & $-1$ & $1$ & $-1$ & $-1$ & $-1$ & $1$ & $-1$ & $1$ & $-1$ & $-1$\\
\hline
$ZZ$ &  $1$ & $1$ & $1$ & $1$ & $-1$ & $-1$ & $1$ & $-1$ & $-1$ & $1$ & $-1$ & $-1$ & $-1$ & $-1$ & $1$ & $1$\\
\hline
\end{tabular}
\caption{Values for the character function $\chi_{\sigma}(P)$ for $P\in \gr{P}_2$ and $\sigma \in \{(Z \id )/2,(\id Z )/2,(Z Z )/2\}$, suppressing the tensor product.}\label{box:char}
\end{table}

\noindent{\bf Interleaved experiment}\\

To perform the interleaved stage of two-for-one interleaved benchmarking we must perform the following set of steps
\begin{enumerate}
\item Choose $\gr{G} = \gr{C}_1\tn{2}$ the group of single qubit Cliffords on two qubits and choose $\gr{\hat{G}} = \gr{P}_2$ the two qubit Pauli group
\item Choose $\{Q,\id- Q\}$ a two component POVM with $Q= \dens{00}$ and choose $\rho = \dens{00}$
\item Choose $\phi_{\sigma}$ with $\sigma = (Z\otimes \id )/2$ an irreducible subrepresentation of the PTM representation of $\gr{P}_2$ with character function $\chi_{\sigma}$ (given explicitly in \cref{box:char})
\item Perform an interleaved character randomized benchmarking experiment (as given in \cref{def:int_char_randomized benchmarking}) $(\gr{C}_1\tn{2},\gr{P}_2, w, \md{M},C)$ (for suitably chosen $\md{M}$) with $\phi = \phi_\sigma$ to obtain the quality parameter $f_w$ with $w= (1,0)$
\item Choose $\phi_{\sigma}$ with $\sigma = (\id\otimes Z )/2$ an irreducible subrepresentation of the PTM representation of $\gr{P}_2$ with character function $\chi_{\sigma}$ (given explicitly in \cref{box:char})
\item Perform an interleaved character randomized benchmarking experiment $(\gr{C}_1\tn{2},\gr{P}_2, w, \md{M},C)$ (for suitably chosen $\md{M}$) with $\phi = \phi_\sigma$ to obtain the quality parameter $f_w$ with $w= (0,1)$
\item Choose $\phi_{\sigma}$ with $\sigma = (Z\otimes Z )/2$ an irreducible subrepresentation of the PTM representation of $\gr{P}_2$ with character function $\chi_{\sigma}$ (given explicitly in \cref{box:char})
\item Perform an interleaved character randomized benchmarking experiment $(\gr{C}_1\tn{2},\gr{P}_2, w, \md{M},C)$ (for suitably chosen $\md{M}$) with $\phi = \phi_\sigma$ to obtain the quality parameter $f_w$ with $w= (1,1)$
\end{enumerate}
Knowing that $f_w = 1$ for $w=(0,0)$ (assuming the noise affecting the gates is CPTP) we can use \cref{lem:trace_to_quality,lem:trace_to_fid} to obtain the average interleaved fidelity $F_{\mathrm{avg}}^{\mathrm{{ref}}}$ as
\begin{equation}
F_{\mathrm{avg}}^{\mathrm{{int}}} = \frac{1}{5}\left(\frac{1}{4}\left(1 + 3 f_{(0,1)} + 3f_{(1,0)} + 9 f_{(1,1)}\right) + 1\right).
\end{equation}

\noindent{\bf Obtaining the gate average fidelity}\\

Given values for $F_{\mathrm{avg}}^{\mathrm{{ref}}}$ and $F_{\mathrm{avg}}^{\mathrm{{int}}}$ (estimated by the protocols above) we can place upper and lower bounds on the average fidelity $F_{\mathrm{avg}}(\mc{\widetilde{C}}\mc{C}\ct)$ of the gate $C$. We will use the optimal bounds derived in~\cite{dugas2016efficiently} which state that
\begin{equation}\label{eq:int_ref_bound}
|\psi^{(\mathrm{int})} - \psi^{(C)}\psi^{(\mathrm{ref})} + (1-\psi^{(C)})(1-\psi^{(\mathrm{ref})})| \leq \sqrt{\psi^{(C)}(1-\psi^{(C)})}\sqrt{\psi^{(\mathrm{ref})}(1-\psi^{(\mathrm{ref})})}
\end{equation}
where 
\begin{equation}
\psi^{(C)} = 2^{-q}((2^q-1)F_{\mathrm{avg}}(\mc{\widetilde{C}}\mc{C}\ct) -1)
\end{equation}
and similarly for $\psi^{(\mathrm{int})}$ and $\psi^{(\mathrm{ref})}$. We can numerically solve the above inequality to obtain lower and upper bounds on the value for $\psi^{(C)}$ given $\psi^{(\mathrm{int})}$ and $\psi^{(\mathrm{ref})}$ and thus for $F_{\mathrm{avg}}(\mc{\widetilde{C}}\mc{C}\ct)$ given $F_{\mathrm{avg}}^{\mathrm{{ref}}}$ and $F_{\mathrm{avg}}^{\mathrm{{int}}}$.\\

An often quoted number for the gate average fidelity $F_{\mathrm{avg}}(\mc{\widetilde{C}}\mc{C}\ct)$ is the `interleaved gate fidelity estimate' $F^{\mathrm{est}}$, given by~\cite{Magesan_2012_interleaved}
\begin{equation}
F^{\mathrm{est}} = 1- \frac{(2^q-1)}{2^q}\left(1-\frac{2^qF_{\mathrm{avg}}^{\mathrm{{ref}}} - 1}{2^q F_{\mathrm{avg}}^{\mathrm{{int}}} -1}\right)
\end{equation}
which can also be estimated using 2-for-1 interleaved benchmarking. We however stress that this number, without further knowledge of the underlying noise process, has no interpretation as a point estimate of $F_{\mathrm{avg}}(\mc{\widetilde{C}}\mc{C}\ct)$ (apart from being a point in the interval given by solving \cref{eq:int_ref_bound}).\\

\noindent{\bf Comparing standard interleaved randomized benchmarking and 2-for-1 interleaved randomized benchmarking}\\

Note that in \cref{eq:int_ref_bound} higher values for  $F_{\mathrm{avg}}^{\mathrm{{ref}}}$ and $F_{\mathrm{avg}}^{\mathrm{{int}}}$ lead to sharper bounds on $F_{\mathrm{avg}}(\mc{\widetilde{C}}\mc{C}\ct)$. This is, apart from lower resource cost, the main advantage of 2-for-1 character randomized benchmarking. In a typical quantum computing platform the single qubit gate fidelity is much higher than the two qubit gate fidelity. Since a typical $2$ qubit Clifford gate is composed of two layers of single qubit gates and a single two qubit gate~\cite{corcoles2013process} the expected reference fidelity in 2-for-1 interleaved randomized benchmarking is much higher than the reference fidelity in standard interleaved randomized benchmarking, thus leading to much sharper bounds on the average fidelity of the interleaved gate. To illustrate this we have simulated 2-for-1 interleaved randomized benchmarking and standard interleaved randomized benchmarking using realistic values for single qubit gate fidelities and two qubit gate fidelities~\cite{watson2018programmable}. In particular we have chosen the single qubit average gate fidelity to be $F_{\mathrm{avg}}^{(1)} = 0.99$ and the two qubit gate fidelity to be $F^{(2)}_{\mathrm{avg}} = 0.898$. In \cref{fig:simulation} we show the result of a simulated experiment using these values. We see that the reference fidelity in 2-for-1 interleaved benchmarking is significantly higher ($F_{\mathrm{avg}}^{\mathrm{ref}} \approx 0.98$) than the reference fidelity of standard interleaved benchmarking ($F_{\mathrm{avg}}^{\mathrm{ref}} \approx 0.87$). This in turn leads to a significantly higher lower bound for the average fidelity of the interleaved gate ($F_{\mathrm{avg}}(\mc{\widetilde{C}}\mc{C}\ct)\gtrsim 0.79$ for 2-for-1 interleaved benchmarking and $F_{\mathrm{avg}}(\mc{\widetilde{C}}\mc{C}\ct)\gtrsim 0.62$ for standard interleaved benchmarking). 
\begin{figure}
\hspace*{-0.7cm}
\includegraphics[scale=0.5]{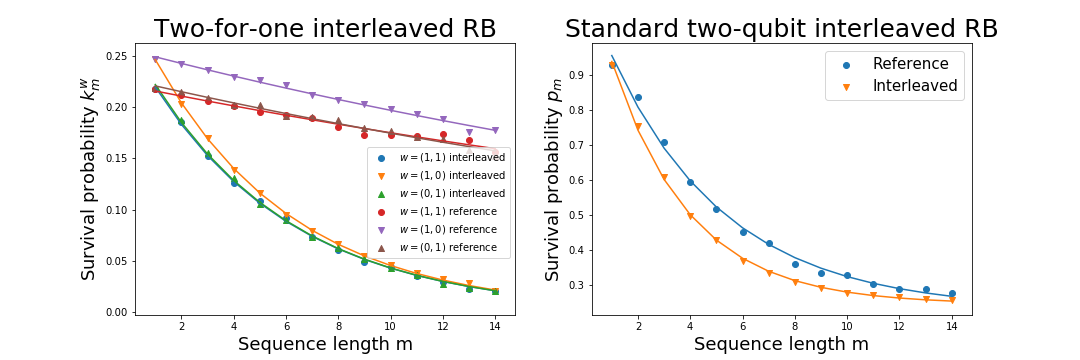}
\caption{Simulation of 2-for-1 interleaved randomized benchmarking (left) and standard two-qubit interleaved randomized benchmarking (right). Inspired by the experimental results of~\cite{watson2018programmable} we chose single qubit gate average fidelities of $F_{\mathrm{avg}} = 0.987$ (on both qubits) and two-qubit gate average fidelities of $F_{\mathrm{avg}}= 0.898$, explicitly realized by a random unitary error map (corresponding to an error model dominated by calibration errors). Also following ~\cite{watson2018programmable} we simulated a measurement fidelity of $F = 0.8$ and a state preparation fidelity of $F= 0.99$. Both experiments sampled 100 random sequences per sequence length for sequence lengths in the interval $[1:15]$. The 2-for-1 interleaved experiment produces a reference fidelity of $F_{\mathrm{ref}}\approx0.98$ and an interleaved fidelity of $F_{\mathrm{int}}\approx0.87$. This leads to an 'interleaved gate fidelity estimate' of $F_{\mathrm{est}}=0.89$ with a guaranteed lower bound of  $F_{\mathrm{avg}}(\mc{\widetilde{C}}\mc{C}\ct)\gtrsim 0.79$.
 On the other hand the standard interleaved randomized benchmarking experiment produces a reference fidelity of $F_{\mathrm{ref}} \approx 0.86$ and an interleaved fidelity of $F_{\mathrm{int}}=0.78$. This leads to an 'interleaved gate fidelity estimate' of $F_{\mathrm{est}}\approx0.9$ with a guaranteed lower bound of $F_{\mathrm{avg}}(\mc{\widetilde{C}}\mc{C}\ct)\gtrsim 0.62$. Note that the lower bound produced by the standard interleaved randomized benchmarking experiment is significantly worse than the lower bound produced by 2-for-1 interleaved benchmarking. (Note that that we have not included error estimates for the fitted values as we are only interested in the qualitative behavior of the experiment here.)  }\label{fig:simulation}
\end{figure}

\section{Supplementary Methods VI: Finite sampling}\label{sec:finite_sampling}
In this section we elaborate on the statistical aspects of character randomized benchmarking. We will denote probability distributions by capital Greek letters (such as $\Lambda$) and their means by the letter $\mu$ subscripted with the corresponding distribution. The character randomized benchmarking protocol requires one to calculate the means of probability distributions. This is however impossible to do exactly using only a finite amount of samples drawn from the probability distribution. Instead one must rely on empirical estimates of these means. The reliability of these estimates is expressed by \emph{confidence intervals}. Imagine being given a distribution with mean $\mu$ and an empirical estimate $\mu_N = \frac{1}{N}\sum_{x\in R_N}x$ where $R_N$ is a set of $N$ samples drawn independently from the distribution.  Now a confidence interval (around $\mu_N$) is a pair of real numbers $(\epsilon, \delta)$ such that 
\begin{equation}
\mathrm{Pr}(|\mu_N-\mu|\geq \epsilon)\leq 1-\delta,
\end{equation}
where the probability is taken with respect to the distribution being sampled from. Even though confidence intervals seem to require knowledge of the distribution being sampled from they can in fact be constructed using only very limited knowledge of the distribution. In particular, if one knows that the distribution being sampled from is bounded, that is it only takes value inside an interval $[a,b]$ for $a,b\in \md{R}$ then we can use Hoeffding's concentration inequality~\cite{Hoeffding1963}, given by
\begin{equation}\label{eq:concentration}
\mathrm{Pr}(|\mu_N-\mu|\geq \epsilon)\leq 1-2\exp\left(\frac{-N\epsilon^2}{(a-b)^2}\right).
\end{equation}
Plugging in $\delta$ and inverting this equation we get a relation between the confidence interval $(\epsilon, \delta)$ and the number of samples $N$ from the distribution we need to construct this interval. We have
\begin{equation}
N\geq \frac{\log(2/\delta)(a-b)^2}{\epsilon^2}.
\end{equation}
Note that this equation is completely generic, it can be used to empirically estimate the mean of any probability distribution, as long as this distribution is bounded.\\

With the above we can analyze the character randomized benchmarking protocol for finite sampling. The main question we aim to answer here is how many samples are required to accurately estimate the character average $k_m^{\lambda}$ for fixed $m$ and $\lambda$. There are $3$ sources of randomness in the character randomized benchmarking protocol.
\begin{enumerate}
	\item The first source of randomness comes from sampling sequences uniformly at random from the set $\gr{G}^{\times m}$
	\item The second source of randomness comes from sampling an element from $\gr{\hat{G}}$ uniformly at random.
	\item The last source of randomness is quantum mechanics itself. In general we can perform the following sequence of events
	\begin{enumerate}
		\item Prepare a system in a state $\rho$
		\item Apply some quantum operation $\mc{E}$
		\item Measure using some two-component POVM $\{Q, \id - Q\}$
	\end{enumerate}
	At the end of this sequence we will get a single bit of information $x$ which takes the value $0$ (measure $Q$) or $1$ (measure $\id -Q$). We can think of $x$ as the being an instance of a random variable $X$ which follows a Bernoulli distribution $\Lambda_{\mathrm{Bern}}$ with mean $ \mu_{\Lambda_{\mathrm{Bern}}}=\bra{Q}\mc{E}\ket{\rho}$.
\end{enumerate}
As mentioned in the main text, one of the key challenges of character randomized benchmarking lies in estimating the mean of the distribution induced by uniform random sampling from the group $\gr{\hat{G}}$ (the second source of randomness). Formally we have
\begin{equation}
k_m^{\lambda}(\vec{G}) = \avg_{\hat{G}\in \gr{\hat{G}}} {\chi_{\hat{\phi}}(\hat{G})}|{\hat{\phi}}|\bra{Q}\widetilde{\mc{G}}_{\mathrm{inv}} \widetilde{\vec{\mc{G}}}\hat{\mc{G}}\ket{\rho}.
\end{equation}
Note that this quantity mixes two of the above types of randomness as $k_m^{\lambda}(\vec{G})$ is an average of quantities $\bra{Q}\widetilde{\mc{G}}_{\mathrm{inv}} \widetilde{\vec{\mc{G}}}\hat{\mc{G}}\ket{\rho}$ which are themselves means of Bernoulli distributions.\\

The naive way of estimating $k_m^{\lambda}(\vec{G})$ would be to first estimate the means $\bra{Q}\widetilde{\mc{G}}_{\mathrm{inv}} \widetilde{\vec{\mc{G}}}\hat{\mc{G}}\ket{\rho}$ by performing the associated measurement procedure $N$ times and using the concentration inequality given above to construct an (accurate) estimate of $\bra{Q}\widetilde{\mc{G}}_{\mathrm{inv}} \widetilde{\vec{\mc{G}}}\hat{\mc{G}}\ket{\rho}$. We can then multiply each estimate by $\chi_{\hat{\phi}}(\hat{G})|{\hat{\phi}}|$ and average them to obtain an estimate for $k_m^\lambda(\vec{G})$.\\

 However, to calculate $k_m^\lambda(\vec{G})$ we would have to perform this procedure for every $\hat{G}\in \gr{\hat{G}}$, which would require $|\gr{\hat{G}}|N$ samples in total. This is not a good approach when performing character randomized benchmarking on more than a few qubits. The reason for this is that typically the size of $\gr{\hat{G}}$ will grow exponentially with the number of qubits. For instance, if $\gr{\hat{G}}$ is the Pauli group we have $|\gr{\hat{G}}|=|\gr{P}| = 4^q$ for $q$ qubits.\\

A second method, which will be more efficient when $|\gr{\hat{G}}|$ is very big, is to not try to estimate all means $\bra{Q}\widetilde{\mc{G}}_{\mathrm{inv}} \widetilde{\vec{\mc{G}}}\hat{\mc{G}}\ket{\rho}$ individually. Instead we will perform an empirical estimate of $k_m^{\lambda}(\vec{G})$ directly by the following procedure.
\begin{enumerate}[leftmargin=*]
	\setlength\itemsep{-0.2em}
	\item Sample $\hat{G}\in \gr{\hat{G}}$ uniformly at random 
	\item Prepare the state $\mc{G}_{\mathrm{inv}}\mc{G}_m\cdots \mc{G}_1\mc{\hat{G}}\ket{\rho}\vspace{0.5mm}$ and measure it once obtaining a result $b(\hat{G}) \in\{0,1\}$
	\item Compute  $x(\hat{G})~\!=~\!{\chi}_{\hat{\phi}}(\hat{G})|{\hat{\phi}}|b(\hat{G})~\!\in~\!\{0,{\chi}_{\hat{\phi}}(\hat{G})|{\hat{\phi}}|\}$
	\item Repeat sufficiently many times and compute the empirical average of $x(\hat{G})$
\end{enumerate}
Every time we perform steps (1)-(3) we are are drawing a single sample from a certain probability distribution. This probability distribution is a \emph{mixture distribution}. Mixture distributions are defined as linear combinations of probability distributions. Note that there there is a difference between a mixture of distributions and an linear combination of random variables~\cite{titterington1985statistical}. Formally the mixture distribution induced by the procedure outlined above will be defined as
\begin{equation}
\Lambda_\lambda = \avg_{\hat{G}\in \gr{\hat{G}}} |{\hat{\phi}}|{\chi_{\hat{\phi}}(\hat{G})}\Lambda_{\mathrm{Bern}, \hat{G}}
\end{equation}
where $\Lambda_{\mathrm{Bern}, \hat{G}}$ is a Bernoulli distribution with mean $\mu_{\Lambda_{\mathrm{Bern}, \hat{G}}}=\bra{Q}\widetilde{\mc{G}}_{\mathrm{inv}} \widetilde{\vec{\mc{G}}}\hat{\mc{G}}\ket{\rho}$. The distribution $\Lambda_\lambda$ will in general be rather complex (as it is the mixture of $|\gr{\hat{G}}|$ Bernoulli distributions). A useful feature of mixture distributions however, is that their mean is given by the weighted average the means of the mixing distributions with the weights precisely given by the weights in the mixture~\cite{titterington1985statistical}. In particular that means we have for $\mu_{\Lambda_\lambda}$ that 
\begin{align}
\mu_{\Lambda_\lambda} &= \avg_{\hat{G}\in \gr{\hat{G}}} |{\hat{\phi}}|{\chi_{\hat{\phi}}(\hat{G})}\mu_{\Lambda_{\mathrm{Bern}, \hat{G}}}\\
&= \avg_{\hat{G}\in \gr{\hat{G}}} |{\hat{\phi}}|{\chi_{\hat{\phi}}(\hat{G})}\bra{Q}\widetilde{\mc{G}}_{\mathrm{inv}} \widetilde{\vec{\mc{G}}}\hat{\mc{G}}\ket{\rho}\\
&= k^\lambda_m(\vec{G}).
\end{align}

 Moreover the distribution $\Lambda_\lambda$ is upper and lower bounded by $\pm|{\hat{\phi}}|\chi*_\lambda$ where  $\chi^*_\lambda = \max_{\hat{G}}|\chi_{\hat{\phi}}(\hat{G})|$. This means that we can use the concentration inequality \cref{eq:concentration} to bound the number of times we need to sample from $\Lambda_\lambda$ (via the procedure above) in order to estimate $k^\lambda_m(\vec{G})$. Note that the number of samples that need to be taken will now not depend on $|\gr{\hat{G}}|$ at all. \\

 As an illustration consider the follow example. Let $\gr{\hat{G}}$ be the Pauli group $\gr{P}_q$ on $q$ qubits. This group is of size $|\gr{P}|=4^q$. However, as discussed above, the subrepresentations of of the Pauli transfer matrix representation $\mc{P}$ are all of dimension one and are indexed by the normalized Pauli matrices $\sigma\in \{\sigma_0\}\cup \bsq$. Let's perform character randomized benchmarking where $\hat{\phi}=\phi_\sigma$ for some normalized Pauli matrix $\sigma$. Since the representation $\phi_\sigma$ is one dimensional we have $|\phi_\sigma| =1$. Moreover we have that the character $|\chi_{\sigma}(P)|=1$ for all $P\in \gr{P}_q$. This means that the distribution $\Lambda_\sigma$ is upper and lower bounded by $\pm 1 $. If we now want to estimate the mean $k_m^\lambda(\vec{G})$ for a particular sequence $\vec{G}$ we can perform the procedure above to sample from $\Lambda_\sigma$. Using the concentration inequality \cref{eq:concentration} see that for a confidence interval of size $\epsilon =0.02$ and confidence $\delta = 0.99$  around the mean $\mu_{\Lambda_\sigma} = k^\lambda_m(\vec{G})$ we need to draw 
 \begin{equation}
 N \geq \frac{\log(2/0.99)(1- (-1))^2}{0.02^2} = 1769
\end{equation}
samples. Note that this number is both `reasonable' and completely independent of the number of qubits $q$. It is moreover an overestimate which could be easily improved using more knowledge of the underlying probability distribution\\

We make a final note about step (1) in the procedure for estimating $\bra{Q}\mc{E}\ket{\rho}$, that is the preparation of the state $\rho$. It will often be the case that the optimal state for a character randomized benchmarking procedure, is not a pure state but rather represented by a density matrix of high rank. This introduces further experimental difficulties as an experimental setup usually only gives access to pure states (by design). We can overcome this difficulty by realizing that every density matrix $\rho$ can be written as a probability distribution over pure states, that is
\begin{equation}
\rho = \sum_{\psi} p^\rho_\psi\dens{\psi},\;\;\;\;\;\;p^\rho_\psi\geq 0,\;\;\;\;\;\;\sum_{\psi}p^\rho_\psi =1.
\end{equation}
This means that $\bra{Q}\mc{E}\ket{\rho}$ is also the mean of a mixture distribution that takes values in the set $\{0,1\}$ (so the mixture is still a Bernoulli distribution). In particular it is a mixture of Bernoulli distributions with mean $\bra{Q}\mc{E}\ket{\psi}$. This means that in the case of non-pure $\rho$ we can update our sampling procedure to be
\begin{enumerate}[leftmargin=*]
	\setlength\itemsep{-0.2em}
	\item Fix a decomposition $\rho = \sum_{\psi} p^\rho_\psi\dens{\psi}$
	\item Sample $\psi$ according to $\{ p^\rho_\psi\}_\psi$
	\item Sample $\hat{G}\in \gr{\hat{G}}$ uniformly at random 
	\item Prepare the state $\mc{G}_{\mathrm{inv}}\mc{G}_m\cdots \mc{G}_1\mc{\hat{G}}\ket{\psi}\vspace{0.5mm}$ and measure it once obtaining a result $b(\hat{G}) \in\{0,1\}$
	\item Compute$\vspace{1mm}$\\ $x(\hat{G})~\!=~\!{\chi}_{\hat{\phi}}(\hat{G})|{\hat{\phi}}|b(\hat{G})~\!\in~\!\{0,{\chi}_{\hat{\phi}}(\hat{G})|{\hat{\phi}}|
\} \vspace{1mm}$
	\item Repeat sufficiently many times and compute the empirical average of $x(\hat{G})$.
\end{enumerate}
This means we are now sampling from the mixture distribution
\begin{equation}
\Lambda_\lambda = \avg_{\hat{G}\in \gr{\hat{G}}} \sum_{\psi}p^\rho_\psi|{\hat{\phi}}|{\chi_{\hat{\phi}}(\hat{G})}\Lambda_{\mathrm{Bern}, \hat{G},\psi}
\end{equation}
where $\Lambda_{\mathrm{Bern}, \hat{G},\psi}$ is now a Bernoulli distribution with mean $\bra{Q}\widetilde{\mc{G}}_{\mathrm{inv}}\widetilde{\mc{G}}_m\cdots \widetilde{\mc{G}_1\mc{\hat{G}}}\ket{\psi}$. However the same reasoning as above holds and the number of samples (repetitions of the above procedure) required to obtain an estimate for the mean of $\Lambda_\lambda$ still only depends on the interval on which $\Lambda_\lambda$ is defined, yielding no increase in the number of samples needed even when the ideal input state $\rho$ is very non-pure (has high rank).

\end{document}